\documentclass{article}
\usepackage{amsthm}
\usepackage{amsmath}
\usepackage{placeins}
\usepackage[table]{xcolor}
\usepackage{tikz}
\usepackage{hyperref}

\usepackage[margin=1in]{geometry}
\usetikzlibrary{arrows,decorations.markings,patterns,calc, positioning, backgrounds}
\tikzstyle{roundnode} = [circle, draw=black, scale=0.7]
\tikzstyle{edgestyle} = [circle, fill=white, scale=0.8]
\tikzstyle{edgestylesmall} = [circle, fill=white, scale=0.6]
\usepackage{listings}
\DeclareMathOperator{\rank}{rank}
\DeclareMathOperator{\first}{first}

\begin{document}

\newtheorem{theorem}{Theorem}
\newtheorem{claim}{Claim}
\newtheorem{lemma}{Lemma}
\newtheorem{definition}{Definition}

\title{Profile-based optimal stable matchings in the Roommates problem}
\author{Sofia Simola \\ University of Glasgow \\ \textit{sofia.simola@tuwien.ac.at} \and David Manlove \\ University of Glasgow \\ \textit{david.manlove@glasgow.ac.uk}}

\maketitle

\begin{abstract}
The stable roommates problem can admit multiple different stable matchings. We have different criteria for deciding which one is optimal, but computing those is often NP-hard.

We show that the problem of finding generous or rank-maximal stable matchings in an instance of the roommates problem with incomplete lists is NP-hard even when the preference lists are at most length 3. We show that just maximising the number of first choices or minimising the number of last choices is NP-hard with the short preference lists.

We show that the number of $R^{th}$ choices, where $R$ is the minimum-regret of a given instance of SRI, is 2-approximable among all the stable matchings. Additionally, we show that the problem of finding a stable matching that maximises the number of first choices does not admit a constant time approximation algorithm and is W[1]-hard with respect to the number of first choices.

We implement integer programming and constraint programming formulations for the optimality criteria of SRI. We find that constraint programming outperforms integer programming and an earlier answer set programming approach by Erdam et. al. \cite{erdem_answer_2020} for most optimality criteria. Integer programming outperforms constraint programming and answer set programming on the almost stable roommates problem.
\end{abstract}

\section{Introduction}

In the Stable Roommates problem (SR) we have a number of agents we want to match in pairs. Each agent lists every other agent in order of preference. We want to find a matching that is stable, i.e. one in which there is no pair of agents that prefer each other to their assigned partners.

SR may admit multiple stable matchings. Since the problem has multiple practical applications \cite{selected_open_problems}, including roommate allocations, chess tournaments \cite{kujansuu1999stable} and peer-to-peer networking \cite{p2p} we often want the stable matching that is best according to some criteria - for example we might want the agent who is worst off to have as good a partner as possible. Unfortunately, it is NP-hard to find stable matchings satisfying many of these criteria in instances of SR.

In many practical applications, agents may not find every other agent acceptable. Perhaps they do not know more than a subset of other agents, or there may be an agent they do not like. We call the problem where the agents may veto other agents the Stable Roommates Problem with Incomplete lists (SRI). No stable matching may contain a pair of agents that find each other unacceptable.

There are conceivable use cases where the number of acceptable agents is bounded by some integer. It is possible, for example, that the agents were only allowed to rank a small number of other agents. We show that rank-maximal and generous SRI and FC-SRI remain NP-hard even when the preference lists are of length at most 3.

In rank-maximal SRI, we are first trying to find a stable matching maximising the number of first choices, then subject to that, the number of second choices, and so on. Generous is similar: we are first trying to find the stable matching minimising the number of $n^{th}$ choices, where $n$ is the number of agents, then subject to that, the number of $(n-1)^{th}$ choices and so on. In FC-SRI we are trying to find the stable matching that maximises the number of first choices.

Sometimes we are happy with a matching that is good enough. We present a 2-approximation algorithm for minimising the number of $k^{th}$ choices among all the stable matchings, where $k$ is the minimum regret - the ranking of the lowest choice any agent is matched to. We also show that there is no constant-time approximation algorithm for the first-choice-maximal SRI.

Integer Programming (IP) and Constraint Programming (CP) are methods for solving NP-hard problems. They use multiple heuristics and are often able to solve problems substantially faster than a brute-force approach. In IP we model our problem as a system of linear inequalities. These inequalities give us constraints, and we want to maximise some function subject to them \cite{conforti_integer_2014}. In CP we have variables and different constraints that relate these variables to each other and define what kinds of solutions are feasible. A CP solver explores the search tree of different possible solutions \cite{pesant_constraint_2014}.

In this paper we present IP and CP models for SRI with different optimality criteria and evaluate their performance. We find that constraint programming outperforms integer programming and the earlier Answer Set Programming (ASP) formulation by Erdam et al. \cite{erdem_answer_2020} for most optimality criteria. Integer programming vastly outperforms constraint programming and answer set programming on the almost stable roommates problem.

In the next section we formally present our problems and go into more detail about their background. The following sections describe our theoretical results: Section~\ref{sec:np_hardness_proofs} describes our proofs for NP-hardness of rank-maximal, generous and first-choice-maximal SRI with short preference lists, Section~\ref{sec:approximability_results} describes the inapproximability and approximability results. In Section~\ref{sec:cp_ip} we describe our IP and CP models and their implementation, which are evaluated in Section~\ref{sec:cp_ip_eval}. Finally, we present conclusions and ideas for future work in Section~\ref{sec:future}.

\section{Background}

\begin{table}[h]
\centering
\begin{tabular}{ c c c c} 
$a_1$: & $a_2$ & $a_3$ & $a_4$ \\ 
$a_2$: & $a_3$ & $a_1$ & $a_4$ \\ 
$a_3$:  & $a_1$ & $a_2$ & $a_4$ \\ 
$a_4$: & \multicolumn{3}{c}{\textit{Arbitrary}} \\
\end{tabular}
\caption{An instance of SR with no stable matching \cite{gale_college_1962}}
\label{table:sr_nostable}
\end{table}

\subsection{Stable Roommates problem}
The Stable Roommates Problem (SR) was first introduced by Gale and Shapley \cite{gale_college_1962}. They noted that unlike the related stable marriage problem, SR does not always admit a stable matching. An instance where that happens is in Figure~\ref{table:sr_nostable}. If we match the pairs $\{a_1, a_2\}$ and $\{a_3, a_4\}$, $a_2$ and $a_3$ prefer each other to their assigned partners. Similarly, if we match $\{a_1, a_3\}$ and $\{a_2, a_4\}$ $a_1$ and $a_2$ prefer each other to their assigned partners, and if we match $\{a_2, a_3\}$ and $\{a_1, a_4\}$, then $a_1$ and $a_3$ block the matching.

Irving \cite{irving_efficient_1985} provided a polynomial-time algorithm that either finds a stable matching for a given instance of SR or reports that none exists in polynomial time.

Gusfield and Irving \cite{gusfield_book_1989} considered the extension of the stable roommates problem where the agents' preference lists do not need to be complete (SRI). They showed that the earlier algorithm by Irving \cite{irving_efficient_1985} extends naturally to this case. They also proved the important result that for any instance of SRI, the set of agents which are matched is the same across all the possible stable matchings \cite[Theorem 4.5.2]{gusfield_book_1989}.

\subsection{NP-hardness results and approximability}

Every matching can be associated with a cost - the sum of how every agent ranks their partner. In the egalitarian SRI problem we want to find a stable matching that minimises this cost. This optimality-criteria was proven to be NP-hard for an instance of SRI by Feder \cite{feder_fixedpoint_1992}. More recently, rank-maximal and generous SRI and FC-SRI were shown to be NP-hard by Cooper \cite{cooper_phd}. The minimum-regret SRI is an exception: Gusfield and Irving \cite{gusfield_book_1989} described an $O(n^2)$ algorithm, where $n$ is the number of agents.

Abraham et al. \cite{abraham_almost_2005} showed that the problem of finding a matching that has as few blocking pairs as possible is also NP-hard. Chen et al. showed that the problem is W[1]-hard with respect the number of blocking pairs \cite{chen_et_al:LIPIcs:2018:9039}.

There are a few approximation algorithms for the egalitarian stable matching for a given instance of SRI in the literature. Feder \cite{feder1994network} and Gusfield et al. \cite{gusfield_bounded_1992} presented 2-approximation algorithms for egalitarian SRI that were special cases of solving the optimal 2-SAT problem. Teo et al. \cite{teo1997lp} present a 2-approximation algorithm for the more general problem of finding an ``optimal" stable matching for a given instance of SRI. ``Optimal" is an extension of egalitarian where the cost of a matching is not restricted to be the sum of the ranks but is allowed to be any other function on the matchings. The algorithm is based on a linear programming formulation for SRI. The cost-functions on the preferences must satisfy the ``U-shaped"-condition, roughly meaning that they must be first decreasing and then increasing when moving down agents' preference lists.

Finding an egalitarian stable matching for a given instance of SRI whose preference lists are short was found to be NP-hard more recently by Cseh et al.\  \cite{cseh_stable_2019}. They showed that the problem remains NP-hard for instances with preference lists at most length 3. They also gave approximation results for the problem of finding an egalitarian stable matching with short preference lists for a given instance of SRI. Chen et al. \cite{chen_et_al:LIPIcs:2018:9039} showed that the problem is FTP with respect to the egalitarian cost.

\subsection{SRI formulations for different solvers}

Prosser \cite{prosser_constraint_2014} presented a CP formulation and implementation for SRI. The approach did not implement different optimality criteria, although Prosser noted that most randomly generated instances did not have many stable matchings. This led to the idea that generating all the stable matchings could be a feasible strategy for finding an optimal one.

ASP has also been used for SRI and its extension with ties (an agent is allowed to be indifferent between two agents) \cite{erdem_answer_2020}. The paper used ASP to find stable matchings of given instances of SRI or report that none exists. They also wrote formulations for finding the stable matchings satisfying the different optimality criteria. The paper found that in most cases CP outperformed ASP, even though the different optimality criteria were found when using CP by enumerating all the possible stable matchings. They also evaluated an earlier answer set formulation by Amendola \cite{amendola_solving_2018}, which mostly outperformed theirs, but did not implement any optimality criteria.

We are not aware of a previous IP implementation for SRI, but there are formulations for other matching problems in the literature. Rothblum \cite{rothblum_characterization_1992} wrote a linear programming formulation for stable marriage with incomplete lists  - a problem similar to SRI, except that agents are divided into two groups and they only find agents in the other group acceptable - which Abeledo and Rothblum \cite{abeledo_stable_1994} note works for SR as it is written for a (not necessarily bipartite) graph. Gusfield and Irving wrote a formulation for the stable marriage problem (with complete preference lists) \cite{gusfield_book_1989} and Vande Vate \cite{vate_linear_1989} for the optimal stable marriage, optimality having the same meaning as in the roommates case.

Other matching problems have IP implementations in the literature. Kwanashie and Manlove \cite{kwanashie_integer_2014} formulated an IP formulation for the Hospital/Residents problem with ties. Hospital/Residents problem with ties is an extension of the stable marriage problem with ties, where many agents in one group (residents) are matched to one agent in the other group (hospitals). When evaluated on instances with up to 400 residents, the median runtime stayed low, although the mean runtime for the biggest instances was above 300 s. This seems to indicate most instances were easy with a few difficult ones. Podhradsky \cite{podhradsky_lp_2010} created a linear programming formulation for approximating Stable Marriage with ties and incomplete lists, which is NP-hard.

\subsection{Formal definitions}

In this section we formally define our problems. We start with general terms and then describe the optimality criteria. Throughout this section, let $I$ be an arbitrary instance of SRI and $M$ an arbitrary matching.

\subsubsection{General terminology}

\begin{definition}[\cite{gale_college_1962}, \cite{gusfield_book_1989}]
In the \textsc{Stable Roommate Problem} (SR) $n$ agents rank every other agent in a strict order of preference. They must be paired in such a way that there are no blocking pairs (see Definition~\ref{def:blocking_pair}). If we allow agents' preference lists to be incomplete, the problem is called \textsc{Stable Roommate Problem with Incomplete Lists} (SRI). If agent $a_i$'s preference list contains another agent $a_j$, we say $a_i$ finds $a_j$ \emph{acceptable}, otherwise $a_i$ finds $a_j$ \emph{unacceptable}.
\end{definition}

\begin{definition}[\cite{gusfield_book_1989}]
A \emph{matching} in the context of SRI is a partition of a subset of agents into disjoint pairs. If agent $a_i$ is in one of the pairs in the matching, we say that $a_i$ is \emph{matched}. Otherwise, we say $a_i$ is unmatched. Moreover, if agent $a_i$ is in a pair with agent $a_j$, we say that $a_j$ is $a_i$'s \emph{partner}.
\end{definition}

\begin{definition}[\cite{gale_college_1962}]\label{def:blocking_pair}
The agents $\{a_i,a_j\}$ form a blocking pair if (i) each of $a_i$ and $a_j$ find one another acceptable; (ii) either $a_i$ is unmatched or prefers $a_j$ to her partner in $M$; and (iii) either $a_j$ is unmatched or prefers $a_i$ to her partner in $M$.
\end{definition}

\begin{definition}[\cite{gusfield_book_1989}]
Given two agents, $a_i$ and $a_j$, who find each other acceptable, $\rank(a_i,a_j)$ is 1 plus the number of agents that $a_i$ prefers to $a_j$. For example, if $a_3$ has the preference list $[a_1, a_2]$, $\rank(a_3, a_1) = 1$ and $\rank(a_3, a_2) = 2$.
\end{definition}

\begin{definition}
Let $I$ be an arbitrary instance of SRI. Then $A^*$ is the set of the agents that are matched in any stable matching of $I$. By theorem 4.5.2 from \cite{gusfield_book_1989}, these are the same for every stable matching of $I$.
\end{definition}

\begin{definition}[\cite{manlove2013algorithmics}]
Let the profile of $M$ be given by the vector $p(M) = \langle p_1\dots,p_L \rangle$ where $p_k = | \{ a_i\in A^* : \rank(a_i,M(a_i))=k  \}| $ and  $L$ is the maximum length of any agent's preference list. This means that the profile describes the number of first, second, etc. choices in a matching.
\end{definition}

\subsubsection{Optimality criteria}
In this section we formally describe the different optimality criteria. If the reader is interested, Appendix~\ref{app:optimality_examples} has some examples.

\begin{definition}[\cite{manlove2013algorithmics}]
A rank-maximal matching is a stable matching whose profile is lexicographically maximum over all the stable matchings of I. In other words, it first maximises the number first choices, then the number of second choices, and so on.
\end{definition}

\begin{definition}[\cite{manlove2013algorithmics}]
A generous stable matching is a stable matching $M$ whose reverse profile, denoted $p_R(M)$ is lexicographically minimum over all the stable matchings of $I$. In other words, it first minimises the number of $n^{th}$ choices, then the number of  $(n-1)^{th}$ choices, and so on, where $n$ is the number of agents.
\end{definition}

\begin{definition}
A first-choice-maximal stable matching is a stable matching $M$ where \\ $|\{ a_i \in A^* : \rank(a_i,M(a_i))=1 \}|$ is maximised.
\end{definition}

\begin{definition}[\cite{manlove2013algorithmics}]
The cost of a matching is \\ $c(M)=\sum_{a_i\in A^*} \rank(a_i,M(a_i))$, In other words, this is the sum of the ranks of the matched agents' partners.
\end{definition}

\begin{definition}[\cite{manlove2013algorithmics}]
An egalitarian stable matching is a stable matching $M$ of $I$ which minimises $c(M)$ over all the stable matchings of $I$.
\end{definition}

\begin{definition}\cite{fleiner_efficient_2007}
The regret of a matching is\\ $r(M) = \max_{a_i\in A^*} \rank(a_i,M(a_i))$. It measures the rank of the worst-off agent.
\end{definition}

\begin{definition}\cite{fleiner_efficient_2007}
A minimum-regret stable matching is a stable matching $M$ of $I$ which minimises $r(M)$ over all the stable matchings of $I$.
\end{definition}

\begin{definition}\cite{abraham_almost_2005}
An almost-stable matching is a matching (not necessarily stable) that admits the fewest number of blocking pairs across all the possible matchings of $I$. If $I$ admits a stable matching, then it is also always an almost-stable matching, as it has 0 blocking pairs.
\end{definition}

Almost-stable is a useful criteria in those real-world situations where no stable matching exists, but the agents need to be matched as well as possible.

\section{NP-hardness proofs}\label{sec:np_hardness_proofs}

In this section we present NP-hardness proofs for the problems of finding first-choice-maximal, rank-maximal and generous stable matchings for a given instance of SRI with short preference lists. Most of these proofs are very similar to each other and rely on the reduction visualised in Figure~\ref{fig:gadgets}.

\subsection{FC-SRI and Rank-maximal SRI with short preference lists}
We begin by defining the problem of finding a rank-maximal and first-choice-maximal matching for a given instance of SRI. Then we move on to describe our proof.

\begin{definition}
Let $I$ be an instance of SRI and $K$ an integer. Then \emph{FC-DEC-SRI} is the problem of deciding whether $I$ admits a stable matching with at least $K$ first choices. If the preference lists of $I$ are of at most length 3, we call the problem \emph{3-FC-DEC-SRI}. We call the optimisation version of the problem \emph{FC-SRI} and when preference lists are of at most length 3, \emph{3-FC-SRI}.
\end{definition}

\begin{theorem}\label{thm:1stmax_hard}
3-FC-DEC-SRI is NP-complete. Therefore 3-FC-SRI is NP-hard.
\end{theorem}

Before presenting our proof, we define the problem we will be reducing from in the proof.

\begin{definition}[\cite{karp1972reducibility}]
Let $G = (V, E)$ be an arbitrary graph. If $C$ is a set of vertices, such that for each $\{u, v\} \in E$ either $u \in C$ or $v \in C$, $C$ is called a \emph{vertex cover} of $G$.
\end{definition}

\begin{definition}[\cite{karp1972reducibility}]
Let $G$ be a graph and $k$ a natural number. Let MIN-VC-DEC be the problem of deciding whether $G$ admits a vertex cover of size at most $k$. Similarly, let MIN-VC be the optimisation version of MIN-VC-DEC.
\end{definition}

MIN-VC-DEC is NP-hard, even in cubic graphs (graphs where each vertex is adjacent to exactly 3 other vertices) \cite{garey_simplified_1974}, \cite{maier1977note}.

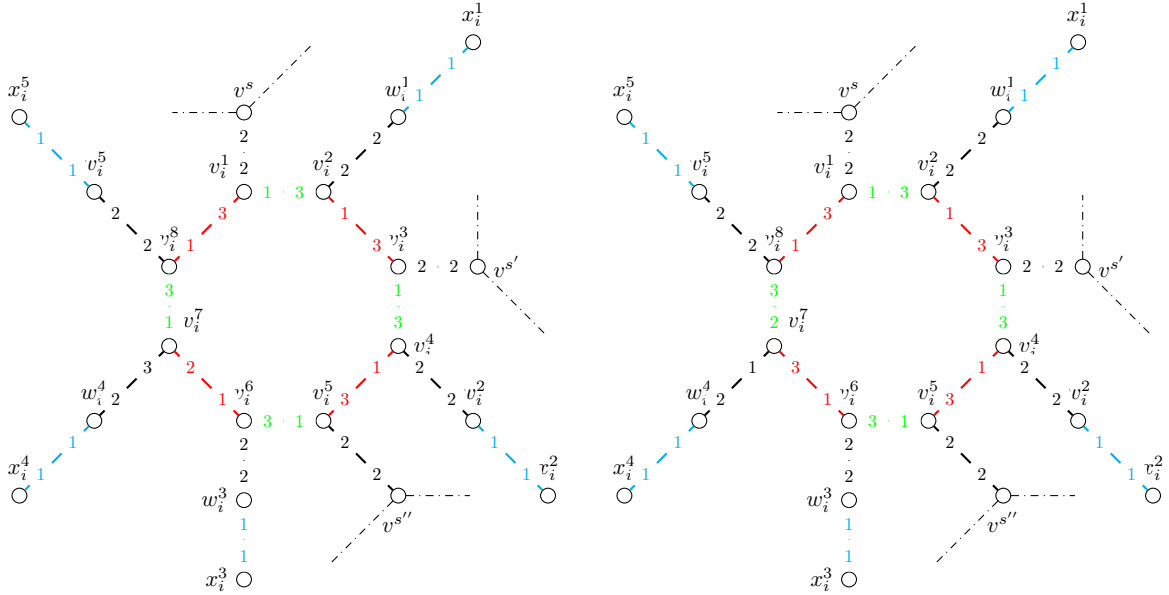
\begin{figure*}
    \centering
\begin{minipage}{0.48\textwidth}
    \begin{tikzpicture}[scale=0.85, transform shape]
]
\pgfmathsetmacro{\d}{2}
\node[roundnode, label=above left:$v^1_i$] (v1) {};
\node[roundnode] (v2) [right=of v1, label=$v^2_i$] {};
\node[roundnode] (v3) [below right=of v2, label=$v^3_i$] {};
\node[roundnode] (v4) [below=of v3, label=right:$v^4_i$] {};
\node[roundnode] (v5) [below left=of v4, label=$v^5_i$] {};
\node[roundnode] (v6) [left=of v5, label=above:$v^6_i$] {};
\node[roundnode] (v7) [above left=of v6, label=above right:$v^7_i$] {};
\node[roundnode] (v8) [above=of v7, label=$v^8_i$] {};

\node[roundnode] (vl) [above=of v1, label=$v^s$] {};
\node[roundnode] (vk) [right=of v3, label=right:$v^{s'}$] {};
\node[roundnode] (vj) [below right=of v5, label=below:$v^{s''}$] {};

\node[] (inv1) [left=of vl] {};
\node[] (inv2) [above=of vk] {};
\node[] (inv3) [below left=of vj] {};
\node[] (inv4) [above right=of vl] {};
\node[] (inv5) [below right=of vk] {};
\node[] (inv6) [right=of vj] {};

\node[roundnode] (w1) [above right=of v2, label=$w^1_i$] {};
\node[roundnode] (x1) [above right=of w1, label=$x^1_i$] {};

\node[roundnode] (w2) [below right=of v4, label=$w^2_i$] {};
\node[roundnode] (x2) [below right=of w2, label=$x^2_i$] {};

\node[roundnode] (w3) [below=of v6, label=left:$w^3_i$] {};
\node[roundnode] (x3) [below=of w3, label=left:$x^3_i$] {};

\node[roundnode] (w4) [below left=of v7, label=$w^4_i$] {};
\node[roundnode] (x4) [below left=of w4, label=$x^4_i$] {};

\node[roundnode] (w5) [above left=of v8, label=$w^5_i$] {};
\node[roundnode] (x5) [above left=of w5, label=$x^5_i$] {};

\draw[thick, green](v1) to node[edgestyle, near start]{1} node[edgestyle, near end]{3} (v2);

\draw[thick, red](v2) to node[edgestyle, near start]{1} node[edgestyle, near end]{3} (v3);

\draw[thick](v3) to node[edgestyle, near start]{2} node[edgestyle, near end]{2}  (vk);

\draw[thick, green](v3) to node[edgestyle, near start]{1} node[edgestyle, near end]{3} (v4);

\draw[thick, red](v4) to node[edgestyle, near start]{1} node[edgestyle, near end]{3} (v5);

\draw[thick](v5) to node[edgestyle, near start]{2} node[edgestyle, near end]{2} (vj);

\draw[thick, green](v5) to node[edgestyle, near start]{1}node[edgestyle, near end]{3}  (v6);

\draw[thick, red](v6) to node[edgestyle, near start]{1} node[edgestyle, near end]{2}  (v7);

\draw[thick, green](v7) to node[edgestyle, near start]{1} node[edgestyle, near end]{3}  (v8);

\draw[thick, red](v8) to node[edgestyle, near start]{1} node[edgestyle, near end]{3} (v1);

\draw[thick](v1) to node[edgestyle, near start]{2} node[edgestyle, near end]{2} (vl);

\draw[thick](v2) to node[edgestyle, near start]{2} node[edgestyle, near end]{2} (w1);
\draw[thick](v4) to node[edgestyle, near start]{2} node[edgestyle, near end]{2}  (w2);
\draw[thick](v6) to node[edgestyle, near start]{2} node[edgestyle, near end]{2}  (w3);
\draw[thick](v7) to node[edgestyle, near start]{3} node[edgestyle, near end]{2} (w4);
\draw[thick](v8) to node[edgestyle, near start]{2} node[edgestyle, near end]{2} (w5);

\draw[thick, cyan](w1) to node[edgestyle, near start]{1} node[edgestyle, near end]{1} (x1);
\draw[thick, cyan, opacity=70](w2) to node[edgestyle, near start]{1} node[edgestyle, near end]{1} (x2);
\draw[thick, cyan, opacity=70](w3) to node[edgestyle, near start]{1} node[edgestyle, near end]{1} (x3);
\draw[thick, cyan, opacity=70](w4) to node[edgestyle, near start]{1} node[edgestyle, near end]{1} (x4);
\draw[thick, cyan, opacity=70](w5) to node[edgestyle, near start]{1} node[edgestyle, near end]{1}  (x5);

\draw[dash dot](vl) to node[]{} (inv1);
\draw[dash dot](vl) to node[]{} (inv4);
\draw[dash dot](vj) to node[]{} (inv3);
\draw[dash dot](vj) to node[]{} (inv6);
\draw[dash dot](vk) to node[]{} (inv2);
\draw[dash dot](vk) to node[]{} (inv5);
\end{tikzpicture}
\end{minipage}
\begin{minipage}{0.48\textwidth}
\begin{tikzpicture}[scale=0.85, transform shape]
]
\pgfmathsetmacro{\d}{2}
\node[roundnode, label=above left:$v^1_i$] (v1) {};
\node[roundnode] (v2) [right=of v1, label=$v^2_i$] {};
\node[roundnode] (v3) [below right=of v2, label=$v^3_i$] {};
\node[roundnode] (v4) [below=of v3, label=right:$v^4_i$] {};
\node[roundnode] (v5) [below left=of v4, label=$v^5_i$] {};
\node[roundnode] (v6) [left=of v5, label=above:$v^6_i$] {};
\node[roundnode] (v7) [above left=of v6, label=above right:$v^7_i$] {};
\node[roundnode] (v8) [above=of v7, label=$v^8_i$] {};

\node[roundnode] (vl) [above=of v1, label=$v^s$] {};
\node[roundnode] (vk) [right=of v3, label=right:$v^{s'}$] {};
\node[roundnode] (vj) [below right=of v5, label=below:$v^{s''}$] {};

\node[] (inv1) [left=of vl] {};
\node[] (inv2) [above=of vk] {};
\node[] (inv3) [below left=of vj] {};
\node[] (inv4) [above right=of vl] {};
\node[] (inv5) [below right=of vk] {};
\node[] (inv6) [right=of vj] {};

\node[roundnode] (w1) [above right=of v2, label=$w^1_i$] {};
\node[roundnode] (x1) [above right=of w1, label=$x^1_i$] {};

\node[roundnode] (w2) [below right=of v4, label=$w^2_i$] {};
\node[roundnode] (x2) [below right=of w2, label=$x^2_i$] {};

\node[roundnode] (w3) [below=of v6, label=left:$w^3_i$] {};
\node[roundnode] (x3) [below=of w3, label=left:$x^3_i$] {};

\node[roundnode] (w4) [below left=of v7, label=$w^4_i$] {};
\node[roundnode] (x4) [below left=of w4, label=$x^4_i$] {};

\node[roundnode] (w5) [above left=of v8, label=$w^5_i$] {};
\node[roundnode] (x5) [above left=of w5, label=$x^5_i$] {};

\draw[thick, green](v1) to node[edgestyle, near start]{1} node[edgestyle, near end]{3} (v2);

\draw[thick, red](v2) to node[edgestyle, near start]{1} node[edgestyle, near end]{3} (v3);

\draw[thick](v3) to node[edgestyle, near start]{2} node[edgestyle, near end]{2}  (vk);

\draw[thick, green](v3) to node[edgestyle, near start]{1} node[edgestyle, near end]{3} (v4);

\draw[thick, red](v4) to node[edgestyle, near start]{1} node[edgestyle, near end]{3} (v5);

\draw[thick](v5) to node[edgestyle, near start]{2} node[edgestyle, near end]{2} (vj);

\draw[thick, green](v5) to node[edgestyle, near start]{1}node[edgestyle, near end]{3}  (v6);

\draw[thick, red](v6) to node[edgestyle, near start]{1} node[edgestyle, near end]{3}  (v7);

\draw[thick, green](v7) to node[edgestyle, near start]{2} node[edgestyle, near end]{3}  (v8);

\draw[thick, red](v8) to node[edgestyle, near start]{1} node[edgestyle, near end]{3} (v1);

\draw[thick](v1) to node[edgestyle, near start]{2} node[edgestyle, near end]{2} (vl);

\draw[thick](v2) to node[edgestyle, near start]{2} node[edgestyle, near end]{2} (w1);
\draw[thick](v4) to node[edgestyle, near start]{2} node[edgestyle, near end]{2}  (w2);
\draw[thick](v6) to node[edgestyle, near start]{2} node[edgestyle, near end]{2}  (w3);
\draw[thick](v7) to node[edgestyle, near start]{1} node[edgestyle, near end]{2} (w4);
\draw[thick](v8) to node[edgestyle, near start]{2} node[edgestyle, near end]{2} (w5);

\draw[thick, cyan](w1) to node[edgestyle, near start]{1} node[edgestyle, near end]{1} (x1);
\draw[thick, cyan](w2) to node[edgestyle, near start]{1} node[edgestyle, near end]{1} (x2);
\draw[thick, cyan](w3) to node[edgestyle, near start]{1} node[edgestyle, near end]{1} (x3);
\draw[thick, cyan](w4) to node[edgestyle, near start]{1} node[edgestyle, near end]{1} (x4);
\draw[thick, cyan](w5) to node[edgestyle, near start]{1} node[edgestyle, near end]{1}  (x5);

\draw[dash dot](vl) to node[]{} (inv1);
\draw[dash dot](vl) to node[]{} (inv4);
\draw[dash dot](vj) to node[]{} (inv3);
\draw[dash dot](vj) to node[]{} (inv6);
\draw[dash dot](vk) to node[]{} (inv2);
\draw[dash dot](vk) to node[]{} (inv5);
\end{tikzpicture}
\end{minipage}
\caption{The preference lists of the agents in Proofs of Theorems~\ref{thm:1stmax_hard}, \ref{thm:gen_hard_general} and \ref{thm:egal3_hard}. The gadget on the left contains preference lists for generous, the gadget on the right for egalitarian and 1st-choice-maximal. Green edges correspond to the case 1 (vertex in the cover) and red edges to the case 2 (vertex not in the cover). The blue edges are present in every stable matching.}\label{fig:gadgets}
\end{figure*}

We will now show 3-FC-DEC-SRI is NP-complere.

\begin{proof}\label{proof:1stmax_hard}
To show the membership in the class NP, we note that we can check the stability of a matching in polynomial time by looking at every potential blocking pair. Counting the number of agents who get their first choice can also be done in polynomial time. We proceed to show the NP-hardness by a reduction from MIN-VC-DEC in cubic graphs:

Let $G = (V,E)$ be an arbitrary cubic graph with vertex set $V = \{u_1, \dots , u_n \}$. We construct the following instance $I$ of SRI with preference lists at most length 3.

We create the sets of agents
\begin{align*}
V' = \{v^r_i : 1 \leq i \leq n \mbox{ and } 1 \leq r \leq 8 \} \\
X = \{x^r_i : 1 \leq i \leq n \mbox{ and } 1 \leq r \leq 5 \} \\
W = \{w^r_i : 1 \leq i \leq n \mbox{ and } 1 \leq r \leq 5 \}
\end{align*}

Let us call the agents corresponding to a vertex $u_i$, namely the agents $\{v^1_i, \dots, v^8_i, x^j_i, w^j_i | 1 \leq j \leq 5 \}$ the \emph{agent group} of $u_i$.

For each $u_i \in U$ give the vertex a counter $c_i$ and initialise it at -1.
For each edge $\{u_i, u_j\}$
\begin{enumerate}
    \item Increment $c_i$ and $c_j$ by 2.
    \item Add $v^{c_i}_i$ and $v^{c_j}_j$ as second choices in each others' preference lists.
\end{enumerate}

Because $G$ is a cubic graph, each $u_i$ is adjacent to 3 other vertices. This means that in the end, each $c^i$ will have the value 5 and each $v^j_i$, where $j \in \{1, 3, 5\}$ will have a second choice in their preference lists.

Their preference lists are in Figure~\ref{fig:preference_lists}, visualised on the right graph on Figure~\ref{fig:gadgets}. The figure is a graph-based representation of SRI: the vertices in the graphs are agents, the edges are acceptable pairs and the labels are preferences. For example the agent $v_i^1$:s first choice is the agent $v_i^2$, because $v_i^1$ and $v_i^2$ are connected by an edge and the label closer to $v_i^1$ is 1. 

\begin{figure}
    \centering
    $v^1_i$: $v^2_i$ $v^s$ $v^8_i$

$v^2_i$: $v^3_i$ $w^1_i$ $v^1_i$

$v^3_i$: $v^4_i$ $v^{s'}$ $v^3_i$

$v^4_i$: $v^5_i$ $w^2_i$ $v^3_i$

$v^5_i$: $v^6_i$ $v^{s''}$ $v^4_i$

$v^6_i$: $v^7_i$ $w^3_i$ $v^5_i$

$v^7_i$: $v^8_i$ $v^6_i$ $w^4_i$ 

$v^8_i$: $v^1_i$ $w^5_i$ $v^7_i$

$w^1_i$: $x^1_i$ $v^2_i$

$w^2_i$: $x^2_i$ $v^4_i$

$w^3_i$: $x^3_i$ $v^6_i$

$w^4_i$: $x^4_i$ $v^7_i$

$w^5_i$: $x^5_i$ $v^8_i$

$x^j_i$ : $w_j$ for $1 \leq j \leq 5$
    \caption{The preference lists of the agents in the proof of Theorem~\ref{thm:1stmax_hard}}
    \label{fig:preference_lists}
\end{figure}

\begin{claim}
If $G$ has a vertex cover $C$ of size $k\leq K$, then $I$ has a stable matching with at least $14n - K$ first choices.
\end{claim}

\begin{proof}
Add $\{x^j_i, w^j_i\}$ to the matching for $1 \leq j \leq 5$. If $u_i \in C$, match the pairs $\{v^1_i, v^2_i\}$, $\{v^3_i, v^4_i\}$, $\{v^5_i, v^6_i\}$ and  $\{v^7_i, v^8_i\}$. If $u_i \notin C$, match the pairs $\{v^2_i, v^3_i\}$, $\{v^4_i, v^5_i\}$, $\{v^6_i, v^7_i\}$ and $\{v^8_i, v^1_i\}$. Let us call the first one \textbf{case 1} and the second one \textbf{case 2}. In both cases, match the agents $w^j_i, x^j_i$ for $1 \leq j \leq 5$.

In case 1, where $u_i$ is in the vertex cover, $v^1_i$, $v^3_i$, $v^5_i$, $x^j_i$ and $w^j_i$ for $1 \leq j \leq 5$ are all matched to their first choices, so they cannot be involved in a blocking pair. $v^2_i$, $v^4_i$ and  $v^8_i$ all have their first and second choices among those and are matched to their second choice, so they cannot be within a blocking pair. $v^6_i$ would prefer $v^7_i$ to their current partner $v^5_i$, but $v^7_i$ prefers $v^8_i$ to $v^6_i$ so they also do not create a blocking pair.

In the case 2 where $u_i$ is not in the vertex cover, $v^2_i$, $v^4_i$, $v^6_i$, $v^8_i$, $x^j_i$ and $w^j_i$ are all matched to their first choices, so they cannot be involved in a blocking pair. $v^7_i$ would prefer $w_i$ and $v^8_i$ to their partner $v^6_i$, but both of them are matched to their first choices. The only possible blocking pairs are $(v^1_i, v^s_j)$, $(v^3_i, v^{s'}_k)$ and $(v^5_i, v^{s''}_l)$. But if any of $v^s_j$, $v^{s'}_k$ or $v^{s''}_l$ is not matched to their first choice, then by construction $u_j$, $u_k$ or $u_l$ is not in the vertex cover, yet there is an edge between $u_i$ and the vertex in question. This would mean $C$ is not a vertex cover, a contradiction. So the matching we have created is stable.

We can see that if a vertex is in the cover, 14 of the corresponding agents are matched to their first choices. If the vertex is not in the cover, 13 of the corresponding agents are matched to their first choices. If the vertex cover has $k$ vertices, then the matching has $14(n-k) + 13k = 14n - k$ first choices. Because $K \geq k$, $14n - k \geq 14n -K$, so the matching has at least $14n - K$ first choices.
\qed \end{proof}

\begin{claim}
If $I$ admits a stable matching of $M$ with at least $14n - K$ first choices, then $G$ has a vertex cover of size at most $K$.
\end{claim}

\begin{proof}\label{proof:1stmax_hard_claim2}
Let us have a stable matching of $I$, with at least $14n - k$ first choices.

We first show that every agent is matched in any stable matching. We can always match the agents in a given group as in the case 1, and this matching will be stable. By \cite{gusfield_book_1989} the set of matched agents is the same for every stable matching, so we know that every agent must be matched for a matching to be stable.

We wish to show that the only possibly stable ways to match agents are the two cases described in the claim 1. $\{x^j_i, w^j_i\}$ are matched to each other in every stable matching, because they are each other's first choices. If we assume no agent is matched to an agent outside of their agent group, the agents' preferences form an even cyclic graph with eight vertices: the agents $v_i^j$ for $1 \leq j \leq n$. There are two ways to match the agents: \{$\{v^1_i, v^2_i\}$, $\{v^3_i, v^4_i\}$, $\{v^5_i, v^6_i\}$, $\{v^7_i, v^8_i\}$\} and  \{$\{v^2_i, v^3_i\}$, $\{v^4_i, v^5_i\}$, $\{v^6_i, v^7_i\}$, $\{v^8_i, v^1_i\}$\}. These are precisely the two cases from claim 1.

What is left to show is that no matching which matches agents to agents outside of their agent group can be stable.


We cannot have an odd number of vertices connected outside of the agent group, because then we would be left with an odd number of agents to match within the group. (As we have seen, $w^j_i$ and $x^j_i$ are always matched to each other.) If we match two vertices outside of the agent group, we have three cases:

\begin{enumerate}
    \item $v^1_i$ and $v^3_i$ are matched outside
    \item $v^1_i$ and $v^5_i$ are matched outside
    \item $v^3_i$ and $v^5_i$ are matched outside
\end{enumerate}

In the cases 1 and 3, there is only one vertex, either $v^2_i$ or $v^4_i$, between the vertices matched outside. This vertex does not have acceptable agents left, meaning it cannot be matched and the matching cannot be stable.

In the case 2, there are three agents between $v^1_i$ and $v^5_i$. This is an odd number of agents, meaning we cannot match all of them.


Let $C = \{u_i \in V : \{v^1_i, v^2_i\} \in M\}$, i.e $C$ is the set of vertices whose agent group is matched as in case 1.

We need to show $C$ is indeed a vertex cover. Assume, for contradiction, that there is an edge which does not have an endpoint in $C$. Then there are two agent groups of $I$, who have agents which find each other acceptable and are matched as in case 2. Let us call those agents $v^a_i$ and $v^b_j$, where $a, b \in \{1,3,5\}$. Because both agent groups are matched according to case 2, $v^a_i$ and $v^b_j$ are matched to their third choices. As they are each others second choices, they form a blocking pair. This means the matching is not stable, a contradiction.

Next, we need to show that $|C| \leq K$. Let $k = |C|$. Then the number of first choices in $M$ is $10n + 3k + 4(n-k) = 14n - k$. So if we were to assume that $k \geq K$ then $M$ has fewer than $14n - K$ first choices. But we assumed that $M$ has at least $14n - K$ first choices, a contradiction.
\qed \end{proof}

We have now shown that $G$ has a vertex cover $C$ of size $k \leq K$ if and only if $I$ has a stable matching with more than $14n - k$ first choices. Therefore, if we had a polynomial-time algorithm FC-DEC-SRI, we could use that to solve MIN-VC-DEC in polynomial time. This is impossible unless P=NP. Thus, FC-DEC-SRI is NP-complete.
\qed \end{proof}

\begin{definition}
Let $I$ be an instance of SRI. Then we call the problem of finding a rank-maximal stable matching \emph{rank-maximal} SRI.
\end{definition}

\begin{theorem}
Rank-maximal SRI is NP-hard even with preference lists of at most length 3.
\end{theorem}

\begin{proof}
Assume, for contradiction, that rank-maximal SRI is polynomial-time solvable with preference lists at most length 3. Then we can use its algorithm to solve the FC SRI in polynomial time, as a rank-maximal stable matching is always a first-choice-maximal stable matching. But we know that FC-SRI does admit a polynomial-time algorithm unless P=NP. Therefore rank-maximal SRI does not admit a polynomial-time algorithm unless P=NP.
\qed \end{proof}

\subsection{Generous SRI with short preference lists}

\begin{definition}
Let $I$ be an instance of SRI. Then we call the problem of finding a generous stable matching \emph{generous SRI}.
\end{definition}

\begin{theorem}\label{thm:gen_hard_general}
Generous SRI is NP-hard even with preference lists of at most length 3.
\end{theorem}

We present two different proofs for this theorem. The first proof follows closely the NP-hardness proof for egalitarian SRI with short preference lists from \cite[Theorem 1]{cseh_stable_2019}. The second proof uses an approach very similar to the proof of Theorem~\ref{thm:1stmax_hard}.

\subsubsection{The first proof of Theorem~\ref{thm:gen_hard_general}}

\begin{proof}
The proof in \cite{cseh_stable_2019} is a reduction from vertex cover in the cubic graphs. We construct an instance $I$ of SRI with preference lists at most length 3 as described in the proof.

The original proof showed $G$ has a vertex cover $C$ where $|C| \leq K$ if and only if $I$ admits a stable matching with cost at most $7m + 19n + K$. We additionally show that if $G$ has a vertex cover $C$ where $|C| \leq K$, then $p^r(M) \leq \langle K + 3n + m, n - K +m, 8n + 2m \rangle $, where $p^r(M)$ is the reverse profile of $M$. This implies that the problem of deciding whether $I$ admits a stable matching with reverse profile which is lexographically smaller than $\langle p_1, \dots, p_L \rangle$ is NP-complete, where $\langle p_1, \dots, p_L \rangle$ is an arbitrary profile and $L$ is the length of the longest preference list any agent has. The optimisation version of this problem is precisely generous SRI, so the result implies that generous SRI is NP-hard.
\\
 
For \cite[Claim 2]{cseh_stable_2019}, instead of claiming that if ``\textit{$G$ has a vertex cover $C$ such that $|C|=k\leq K$, then there is stable matching $M$ in $I$ such that $c(M) \leq K'$, where $K'=7m+19n+K$}" we claim that if $G$ has a vertex cover $C$ such that $|C|=k\leq K$, then there is stable matching $M$ in $I$ such that $p^r(M) \leq \langle K + 3n + m, n - K +m, 8n + 2m \rangle $.

For a fixed $i$, $V^c_i$ assigns 4 agents to their 1st choices and 4 to their 3rd choices, therefore adding $\langle 4, 0, 4 \rangle$ to the profile. Similarly $V^u_i$ adds $\langle 4, 1, 3 \rangle$ to the profile and $M^Z_i$ $\langle 4,0, 0 \rangle$. For a fixed $j$, both $E^1_j$ and $E^2_j$ add $\langle 2, 1, 1 \rangle$ to the profile.

In the matching $M$ we have $k$ copies of $V^c_j$, $(n-k)$ copies of $V^u_i$ and $m$ copies of $M^i_z$. We also know that there is a $E^1_j$ or $E^2_j$, but not both, for each edge. Therefore the profile is 

\begin{align*}
\langle 4, 0, 4 \rangle k + \langle 4, 1, 3 \rangle (n-k) + \langle 2, 1, 1 \rangle m + \langle 4,0, 0 \rangle n \\
= \langle 4k + 4(n-k) + 2m + 4n, n - k + m, 4k + 3(n-k) + m \rangle \\
= \langle 8n + 2m, n - k + m, 3n + m + k\rangle
\end{align*}

and the reverse profile
\[\langle 3n + m + k, n - k + m, 8n + 2m\rangle \leq \langle 3n + m + K, n - K + m, 8n + 2m\rangle\]

as required.

\cite[Claim 3]{cseh_stable_2019} states that ``\textit{If there is stable matching $M$ in $I$ such that $c(M) \leq K'$, where $K'=7m+19n+K$, then $G$ has a vertex cover $C$ such that $|C|=k\leq K$}". We replace $c(M) \leq K'$ with $p^r(M) \leq \langle 3n + m + K, n - K + m, 8n + 2m\rangle $. The details of the proof of the Claim 3 remain otherwise unchanged, as they rely on the construction which we did not alter.
\qed \end{proof}

\subsubsection{The second proof of Theorem~\ref{thm:gen_hard_general}}

We first define the problem of minimising the number of $R^{th}$ choices for a given instance of SRI, where $R$ is the regret of a minimum regret stable matching for that instance. We show that this problem is NP-hard for short preference lists. Consequently, we show that the generous SRI is NP-hard even when the preference lists are short.

\begin{definition}
Let $I$ be an instance of SRI and $K$ an integer. Let $R$ be the minimum regret of $I$. Then \emph{LC-DEC-SRI} is the problem of deciding whether $I$ admits a stable matching with at most $K$ $R^{th}$ choices and 0 choices higher than $R$. In other words, $p(M)=\langle p_1,\dots,p_L\rangle$ where $p_R\leq K$ and $ p_i=0$ for $R+1\leq i\leq L$ and $L$ is the maximum length of a preference list in $I$.
\end{definition}

\begin{definition}
Let \emph{LC-SRI} be the minimisation version of \emph{LC-DEC-SRI}.
\end{definition}

\begin{theorem}\label{thrm:generous_hard}
LC-DEC-SRI is NP-complete even when preference lists are at most length 3.
\end{theorem}

\begin{proof}\label{proof:generousnphard2}
We show the membership in the class NP similarly to the proof of Theorem~\ref{thm:1stmax_hard}. Similarly to the number of first choices, we can also compute the profile of a matching in polynomial time. We prove
NP-hardness by a similar reduction from the minimum vertex cover in cubic graphs.

Let $G = (V,E)$ be an arbitrary cubic graph with vertex set $V = \{u_1, \dots,  u_n \}$.

We create the same set of agents as in the proof of Theorem~\ref{thm:1stmax_hard}. Their preference lists are identical, except for the agent $v^7_i$, whose preference lists are as follows: \[v^7_i : v^8_i v^6_i w^4_i\] for each $1 \leq i \leq n$. The preference lists are visualised on the left in Figure~\ref{fig:gadgets}. 

\begin{claim}
If $G$ has a vertex cover $C$ of size $k\leq K$, then $I$ has a stable matching with at most $3n + K$ third choices.
\end{claim}

\begin{proof}
Let case 1 and case 2 be as in Figure~\ref{fig:gadgets}.

We have shown in the proof of Theorem~\ref{thm:1stmax_hard} that both cases are stable when the preference lists of $v^7_i$ are $ v^6_i w^4_i v^8_i$. Since we only change $v^7_i$, any blocking pair must involve $v^7_i$. In case 1, $v^7_i$ is matched to their first choice, and cannot hence be involved in a blocking pair. In case 2, $v^7_i$ is matched to their second choice, and their first choice is matched to their first choice. Therefore $v^7_i$ cannot be involved in a blocking pair here either.

We can see that if a vertex is in the cover, four of the corresponding agents are matched to their third choices. If the vertex is not in the cover, three of the corresponding agents are matched to their third choices. If the vertex cover has $k$ vertices, then the matching has $4k +3(n-k) = 3n + k$ third choices. Because $K \leq k$, $3n + k \leq 3n +K$, so the matching has at most $3n + K$ third choices.
\qed \end{proof}

\begin{claim}
If $I$ admits a stable matching $M$ with at most $3n + K$ third choices, then $G$ has a vertex cover of size at most $K$.
\end{claim}

\begin{proof}
Let us have a stable matching $M$ of $I$, with at most $3n + K$ third choices.

We need to show that cases 1 and 2 are the only two stable ways to match agents. We have shown this to hold when the agent $v^7_i$ has different preference lists. Since $v^7_i$ is the only agent that is different, they would need to be matched to someone else in a different case. But the only agent $v^7_i$ is not matched to in cases 1 and 2 is $w^4_i$ who is matched to their first choice in every possible stable matching.

Let $C = \{u_i \in V : \{v^1_i, v^2_i\} \in M\}$, i.e $C$ is the set of vertices whose agent group is matched as in case 1. This is a vertex cover by the logic of the proof of Theorem~\ref{thm:1stmax_hard}.

Next, we need to show that $|C| \leq K$. Let $k = |C|$. Then the number of third choices in $M$ is $4k + 3(n-k) = 3n + k$. So if $k > K$ then $M$ has more than $3n + K$ third choices. But we assumed that $M$ has at most $3n + K$ third choices, a contradiction.
\qed \end{proof}

We have now shown that $G$ has a vertex cover $C$ of size $k \leq K$ if and only if $I$ has a stable matching with at most $3n + K$ third choices. Therefore, if we had a polynomial-time algorithm LC-DEC-SRI, we could use that to solve MIN-VC-DEC in polynomial time. This is impossible unless P=NP. Thus, LC-DEC-SRI is NP-complete.
\qed \end{proof}

Now we can prove Theorem~\ref{thm:gen_hard_general}.

\begin{proof}
Assume, for contradiction, that generous SRI is polynomial-time solvable with preference lists at most length 3. Then we can use its algorithm to solve the LC-SRI in polynomial time, as a generous stable matching is always a stable matching with a minimum number of last choices. But we know that LC-SRI does not admit a polynomial-time algorithm unless P=NP. Therefore generous SRI does not admit a polynomial-time algorithm unless P=NP.
\qed \end{proof}

\subsection{Egalitarian SRI with short preference lists}

\begin{theorem}[Theorem 1 from \cite{cseh_stable_2019}]\label{thm:egal3_hard}
Egalitarian SRI is NP-hard even when the preference lists are of at most length 3.
\end{theorem}

Theorem~\ref{thm:egal3_hard} was proven in \cite{cseh_stable_2019}. We now present an alternative proof similar to the proofs of Theorem~\ref{thm:1stmax_hard} and ~\ref{thm:gen_hard_general}.

We define the decision version of egalitarian SRI as follows.

\begin{definition}
Let $I$ be an arbitrary instance of SRI and $k$ an integer. We call the problem of deciding whether $I$ admits a stable matching $M$ with $c(M) \leq k$ EGAL-SRI-DEC.
\end{definition}

We prove that EGAL-SRI-DEC is NP-complete even when the preference lists are of at most length 3. This implies that egalitarian SRI is NP-hard even when the preference lists are of at most length 3.

\begin{proof}\label{proof:egalitarian}
We show the membership to the class NP similarly to the proof of Theorem~\ref{thm:1stmax_hard}. Similarly to the number of first choices, we can also compute the cost of a matching in polynomial time.

We prove the NP-completeness by the reduction from MIN-VC-DC in cubic graphs we used in the proof of Theorem~\ref{thm:1stmax_hard}, presented in Figure~\ref{fig:gadgets}.

Let $G = (V, E)$ be an arbitrary cubic graph with vertex set $V = \{u_1, \dots , u_n \}$ and $I$ the SRI instance reduced from it. We claim that $G$ has a vertex cover of size at most $K$ if and only if $I$ has a stable matching of cost at most $26n + K$.


\begin{claim}
If $G$ has a vertex cover $C$ of size $k\leq K$, then $I$ has a stable matching with at most cost $26n+K$
\end{claim}

\begin{proof}
The cases 1 and 2 are defined as in proof of Theorem~\ref{thm:1stmax_hard} and presented in Figure~\ref{fig:gadgets}.

The reasoning for cases 1 and 2 being stable is identical to the proof of Theorem~\ref{thm:1stmax_hard}.

We can see that in both cases 1 and 2, 4 agents are matched to their third choices. In case 1 (vertex in the cover), 3 + 10 agents are matched to their first choices and 1 to their second choice. In case 2 (vertex not in the cover), 4+10 agents are matched to their first choices and none to their second choice. Therefore each case 1 adds $4\times 3 +2 + 13 = 27$ and each case 2 adds $4 \times 3 + 14  = 26$ to the cost. If the vertex cover has $k$ vertices, then the matching has $27k +26(n-k) = 26n + k$ third choices. Because $K \leq k$, $26n + k \leq 26n +K$, so the matching has cost at most $26n + K$s.
\qed \end{proof}

\begin{claim}
If $I$ admits a stable matching $M$ with $c(M) \leq 26n + K$, then $G$ has a vertex cover of size at most $K$.
\end{claim}

\begin{proof}
Let us have a stable matching of $I$, with the cost at most $26n + K$.

By the logic from the proof of Theorem~\ref{thm:1stmax_hard}, cases 1 and 2 are the only possible stable ways to match the agents.

Similarly, define $C = \{u_i \in V : \{v^1_i, v^2_i\} \in M\}$, i.e. the vertices that are matched as in the case 1. $C$ is a vertex cover by the logic from the proof of Theorem~\ref{thm:1stmax_hard}.

Next, we need to show that $|C| \leq K$. Let $k = |C|$. Then the cost of the matching is $26n + k$. If we were to have $k > K$ then $M$ would have cost $26n + k > 26n + K$. But we assumed that $M$ has a cost at most $26n + K$, a contradiction. So $k \leq K$.
\qed \end{proof}

 \end{proof}

\section{Approximability results}\label{sec:approximability_results}
In this section we first describe an inapproximability result for FC-SRI and then an approximation algorithm for the LC-SRI.

\subsection{Inapproximability and W[1] hardness of FC-SRI}

We show that first-choice-maximal SRI is inapproximable by any constant factor. This means that we are unlikely to find a meaningful approximation algorithm with a constant performance guarantee for rank-maximal SRI: it is not clear how one would approximate profiles, but since rank-maximal SRI maximises the number of first choices over all the stable matchings one would assume that a meaningful measure would restrict the number of first choices.

We first reduce from maximum independent vertex set to SRI. This reduction is inspired by a reduction from MIN-VC by Cooper et al. \cite{cooper_phd}.

Our reduction can also be used to show that FC-SRI is not in FPT with respect to the number of first choices, making an efficient exact algorithm difficult to find. We show an exact polynomial-time-algorithm when the number of first choices is treated as a constant.

\begin{definition}[\cite{tarjan1977finding}]
Let $G$ be a graph. If $C$ is a set of vertices such that no pair of vertices in $C$ are adjacent, $C$ is called an \emph{independent set} of $G$.
\end{definition}

\begin{definition}[\cite{tarjan1977finding}]
Let $G$ be an arbitrary graph and $k$ a natural number. Let MAX-IS-DEC be the problem of deciding whether $G$ admits an independent set of size at least $k$. Similarly, let MAX-IS be the optimisation version of MAX-IS-DEC, i.e. $k$ is omitted from the problem input and we are looking for the biggest $k$ such that $G$ admits an independent set of size $k$.
\end{definition}

MAX-IS on general graphs is NP-hard. Moreover, it is not approximable by any constant factor \cite{hstad1996clique}.

For an instance $G = (V, E)$, where $V = \{u_1, \dots, u_n\}$, construct an instance $I$ of SRI with agents 
\begin{align*}
V' = \{v_i : 1 \leq i \leq n\} \\
W = \{w_i : 1 \leq i \leq n\} \\
X = \{x_i : 1 \leq i \leq n\} \\
Y = \{y_i : 1 \leq i \leq n\} \\
A = \{a_i : 1 \leq i \leq n\} \\
B = \{b_i : 1 \leq i \leq n\}
\end{align*}

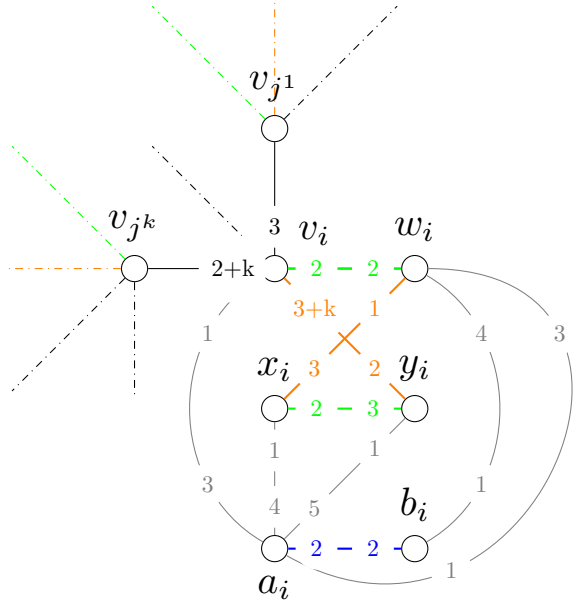
\begin{figure}
    \centering
    \begin{tikzpicture}[scale=1.5, transform shape]
]
\pgfmathsetmacro{\d}{2}
\node[roundnode, label=above right:$v_i$] (v) {};
\node[roundnode] (w) [right=of v,label=$w_i$] {};
\node[roundnode] (x) [below=of v,label=$x_i$] {};
\node[roundnode] (y) [right=of x,label=$y_i$] {};
\node[roundnode] (a) [below=of x,label=below:$a_i$] {};
\node[roundnode] (b) [right=of a,label=$b_i$] {};

\node[roundnode] (v1) [above=of v,label=$v_{j^1}$] {};
\node[] (vinv) [above left=of v] {};
\node[roundnode] (vk) [left=of v,label=$v_{j^k}$] {};

\node[] (inv1) [above left=of v1] {};
\node[] (inv2) [above=of v1] {};
\node[] (inv3) [above right=of v1] {};
\node[] (inv4) [above left=of vk] {};
\node[] (inv5) [left=of vk] {};
\node[] (inv6) [below left=of vk] {};
\node[] (inv7) [below=of vk] {};

\draw[gray](v) to[out=210,in=150,distance=1cm] node[edgestylesmall, near start]{1} node[edgestylesmall, near end]{3} (a);
\draw[thick, green](v) to node[edgestylesmall, near start]{2} node[edgestylesmall, near end]{3} (w);
\draw[](v) to node[edgestylesmall, near start]{3} (v1);
\draw[dash dot](v) to node[]{} (vinv);
\draw[](v) to node[edgestylesmall, near start]{2+k} (vk);
\draw[thick, orange](v) to node[edgestylesmall, near start]{3+k} node[edgestylesmall, near end]{2} (y);

\draw[thick, orange](w) to node[edgestylesmall, near start]{1} node[edgestylesmall, near end]{3} (x);
\draw[thick, green](w) to node[edgestylesmall, near start]{2} node[edgestylesmall, near end]{2} (v);
\draw[gray](w) to[out=0,in=-30,distance=2.5cm] node[edgestylesmall, near start]{3} node[edgestylesmall, near end]{1} (a);
\draw[gray](w) to[out=-30,in=30,distance=1cm] node[edgestylesmall, near start]{4} node[edgestylesmall, near end]{1} (b);

\draw[gray](x) to node[edgestylesmall, near start]{1} node[edgestylesmall, near end]{4} (a);
\draw[thick, green](x) to node[edgestylesmall, near start]{2} node[edgestylesmall, near end]{3} (y);

\draw[gray](y) to node[edgestylesmall, near start]{1} node[edgestylesmall, near end]{5} (a);

\draw[thick, blue](a) to node[edgestylesmall, near start]{2} node[edgestylesmall, near end]{2} (b);

\draw[dash dot,green](v1) to node[]{} (inv1);
\draw[dash dot,orange](v1) to node[]{} (inv2);
\draw[dash dot](v1) to node[]{} (inv3);
\draw[dash dot,green](vk) to node[]{} (inv4);
\draw[dash dot,orange](vk) to node[]{} (inv5);
\draw[dash dot](vk) to node[]{} (inv6);
\draw[dash dot](vk) to node[]{} (inv7);

\end{tikzpicture}
\caption{The preference lists for the proof of Theorem~\ref{thm:1stmax-nonapprox}. The grey edges are not present in the reduced preference list and hence not in any stable matching.}\label{fig:nonapprox_graph}
\end{figure}

These agents have the following preferences:\\
$v_i$ : $a_i$ $w_i$ [$v_j$ for $1 \leq j \leq n $ if $\{u_i, u_j\} \in E$] $y_i$ \\
$w_i$ : $x_i$ $v_i$ $a_i$ $b_i$ \\
$x_i$ : $a_i$ $y_i$ $ w_i$\\
$y_i$: $a_i$ $v_i$ $x_i$\\
$a_i$ : $w_i$ $b_i$ $v_i$ $x_i$ $y_i$ \\
$b_i$ : $w_i$ $a_i$ \\

We visualise them in Figure~\ref{fig:nonapprox_graph}.
\begin{claim}\label{claim:nonapprox1}
If $G$ has an independent vertex set $S$ of size $k \geq K$ then $I$ has a stable matching $M$ with at least $K$ first choices.
\end{claim}

\begin{proof}
Note that $a_i$'s first choice is $w_i$, so there is no stable matching where $w_i$ will be matched to someone they like less than $a_i$. Therefore we can remove the pair $\{b_i, w_i\}$ from the preference lists. Now $b_i$'s first choice is $a_i$, so $a_i$ cannot be matched to anyone they like less than $b_i$. This removes the pairs $\{a_i, v_i\}$, $\{a_i, x_i\}$ and $\{a_i, y_i\}$. Now $v_i$'s first choice is $w_i$, meaning that there is no stable matching, where $w_i$ matched to someone they like less than $v_i$. This removes the pair $\{a_i, w_i\}$ from the preference lists.

We are left with the reduced table:\\
$v_i$ : $w_i$ [$v_j$ for $1 \leq j \leq n $ if $\{u_i, u_j\} \in E$] $y_i$ \\
$w_i$ : $x_i$ $v_i$\\
$x_i$ : $y_i$ $ w_i$\\
$y_i$: $v_i$ $x_i$\\
$a_i$ : $b_i$ \\
$b_i$ : $a_i$ \\

Any pairs that are not on this list cannot be stable.

If $u_i \notin S$, add the pairs $\{v_i, w_i\}$ and $\{x_i, y_i\}$ to $M$ - 0 first choices. If $u_i \in S$, add the pairs $\{v_i, y_i\}$ and $\{w_i, x_i\}$ to $M$- 1 first choice. In any case, add the pair $\{a_i, b_i\}$.

This matching is stable because if $u_i \notin S$
\begin{itemize}
\setlength\itemsep{0em}
    \item $v_i$ is matched to their (reduced) first  choice
    \item $w_i$'s first choice $x_i$ is matched to their first choice and $w_i$ is matched to their second choice
    \item $x_i$ is matched to their (reduced) first choice
    \item $y_i$'s (reduced) first choice $v_i$ is matched to their (reduced) first choice and $y_i$ is matched to their (reduced) second choice
\end{itemize}
and if $u_i \in S$
\begin{itemize}
\setlength\itemsep{0em}
    \item $w_i$ is matched to their first choice $x_i$
    \item $y_i$ is matched to their (reduced) first choice
    \item $x_i$'s (reduced) first choice $y_i$ is matched to their (reduced) first choice
    \item $v_i's$ (reduced) first choice has their first choice, and if any $v_j$ in $v_i$'s preference list is not matched to their first choice, $u_i, u_j \in S$ but $\{u_i, u_j\} \in E$, a contradiction to $S$ being an independent set.
\end{itemize}

We can see that for each vertex in $S$ we get one (non-reduced) first choice, and for each vertex not in $S$, we get 0 (non-reduced) first choices. Therefore $M$ has at least $K$ first choices if $S$ has $k \geq K$ vertices.
\end{proof}

\begin{claim}\label{claim:nonapprox2}
Let $M$ be a stable matching with at least $K$ first choices. Then $G$ has an independent vertex set $S$ of at least the size $K$.
\end{claim}

\begin{proof}
Since we can match every agent as in case 1, a stable matching of $M$ exists and matches every agent \cite[Theorem 4.5.2]{gusfield_book_1989}. By looking at the reduced preference lists from the previous claim, it is clear that the two cases are the only possible ways to match every agents, unless some $v_i, v_j$ are paired to each other.

We now show this is impossible. If some $v_i$ is matched to some $v_j$, either $w_i, x_i$ or $y_i$ must be unmatched. But we have shown that any stable matching matches every agent, a contradiction.

Let $S$ consist of the vertices whose agents are matched as in case 2. This is an independent set, because if there was an edge between two agents in it, say $u_i$ and $u_j$, then $v_i$ would be matched to $y_i$ and $v_j$ to $y_j$, meaning $v_i$ and $v_j$ would form a blocking pair. This would contradict the stability of $M$.

Next we need to show that $|S| \geq K$. Let $k = |S|$. Then there are $k$ first choices in $M$. But since we assumed that there are at least $K$ first choices in $M$, $k \geq K$.
\end{proof}

\begin{theorem}\label{thm:1stmax-nonapprox}

FC-SRI does not admit an approximation algorithm with constant performance guarantee unless P=NP.
\end{theorem}

\begin{proof}
Let us reduce a graph $G$ to an instance $I$ of SRI as above. 

Let $l_{opt}$ be the optimal number of first choices in $I$. Then the maximum independent set of $G$ has the size $l_{opt}$: By Claim~\ref{claim:nonapprox2}, $G$ has an independent set of at least the size $l_{opt}$. By Claim~\ref{claim:nonapprox1}, if the maximum independent set of $G$ is  $k > l_{opt}$, then $I$ has a matching with $k > l_{opt}$ first choices, contradicting the optimality of $l_{opt}$.

Assume FC-SRI admits an approximation algorithm with a constant performance guarantee $c$. Then the algorithm will give us a matching with $l_{opt} / c$ first choices. By Claim~\ref{claim:nonapprox2} we can translate this to an independent set of $G$ in polynomial time, with the independent set having at least the size $ l_{opt} / c$. But then we have a c-approximation algorithm for the general case of the maximum independent set, which does not exist unless P = NP \cite{hstad1996clique}.
\end{proof}

We can use the construction of this proof to also show that FC-SRI is W[1]-hard with respect to the number of first choices.

\begin{theorem}\label{thm:w1-hard}
FC-SRI is W[1]-hard with respect to the number of first choices.
\end{theorem}

\begin{proof}
We know that MAX-IS-DEC is W[1]-hard with respect to the size of the independent set\cite{downey1995fixed}. Our polynomial-time reduction from MAX-IS to FC-SRI is a parameterised reduction:
We have already shown that a graph $G$ admits an independent set of size at least $K$ if and only if the reduced instance $I$ of FC-SRI has at least $K$ first choices. Therefore we have a polynomial-time reduction from $(G, K)$ to $(I, K)$ where $(G, K)$ is a yes-instance if and only if $(I, K)$ is a yes-instance. 
\end{proof}

\begin{theorem}
FC-SRI admits an $O(n^{k+2})$-time algorithm, where $n$ is the number of agents and $k$ the number of first choices in the sought stable matching. Therefore FC-SRI is in XP with respect to the number of first choices.
\end{theorem}

\begin{proof}
Let us generate every subset of $n$ of size exactly $k$. There are $n^k$ such subsets. For each subset $U$, let us create a set of forced edges $\{ \{a_i, \first (a_i) \} | a_i \in U \}$, where $\first (a_i)$ is the agent $a_j$ such that $\rank (a_i, a_j) = 1$. We can solve in $O(m)$ time, where $m$ is the number of all the acceptable pairs, whether there is a stable matching with the forced edges present~\cite{fleiner_efficient_2007}. If such a matching exists, we can report this as a solution. If no such matching exists for any subset $U$, we do not have stable matching with at least $k$ first choices.

We have $n^k$ subset, which is in $O(n^k)$. For each of them, the forced-pair algorithm runs in $O(m) = O(n^2)$ time. Therefore the complexity of our algorithm is $O(n^{k+2})$.
\end{proof}

\subsection{2-Approximability of LC-SRI}

We now present a 2-approximation algorithm for LC-SRI. This, although not an approximation algorithm for generous SRI, approximates to most significant value in the reverse profile.

\begin{theorem}
\emph{LC-SRI} admits a 2-approximation algorithm.
\end{theorem}

\begin{proof}
Let $I$ be an instance of SRI with an agent set $A$. We can solve minimum-regret SRI in polynomial time \cite{Knuth_lectures_1976}. Let us solve for this. Throughout this section, let $R$ denote the minimum regret of $I$. We remove from each agent of $I$ any choice whose rank is higher than $R$. We denote this reduced instance as $I'$. Any stable matching of $I'$ must be a minimum-regret stable matching of $I$.

\begin{definition}[\cite{gusfield_book_1989}]
Optimal SRI is an extension of general SRI, where we have some real-valued weight function $c(u,v)$. Given an instance $I$, we aim to find a stable matching $M$ that minimises $\sum_{u,v \in M}c(u,v)$. 
\end{definition}

By \cite{teo1997lp}, the optimal SRI can be 2-approximated, if the associated cost function $c$ satisfies the ``U-shape" condition. Namely,
\begin{quote}[\cite{teo1997lp}]
for each $i$ with preference list $i_i <_i \dots <_i i_n$, there is a node $i_k$ such that \\
$c(i, i_1) \geq \dots \geq c(i, i_{k-1}) \geq c(i, i_k)$ \\
and\\
$c(i, i_k) \leq c(i, i_{k+1} \dots \leq ) \leq c(i, i_n)$. 
\end{quote}
Let our cost function be 
\begin{equation*}
c(i,j) = \begin{cases}
    1 \quad \textit{j is i's $R^{th}$ choice} \\
    0 \quad \textit{otherwise}
\end{cases}
\end{equation*}

This satisfies the U-shaped condition. For any agent $i$, if they find fewer than $R$ agents acceptable, the sequence of costs is $c(i, i_1) \dots c(i_L)$ = $0, 0, \dots 0$, where $L$ is the number of agents they find acceptable. This is trivially U-shaped, as a constant sequence is both decreasing and increasing.

On the other hand, if agent $i$ finds $R$ agents acceptable in $I'$, the sequence of costs is $c(i, i_1) \dots c(i_R)$ = $0, 0, \dots 0, 1$. If we take $k=0$, the proportion before k is empty, hence trivially decreasing, and the proportion after k is increasing.

Let $M$ be a matching. Because $c(i,j) = 1$ if and only if $j$ is $i$'s $R^{th}$ choice, the sum $\sum_{\{i, j\} \in M}c(i,j) + c(j,i)$ is the number of people in $M$ who get their $R^{th}$ choice. As we have a 2-approximation algorithm for this sum over all the possible matchings, we have a 2-approximation algorithm for LC-SRI.
\end{proof}

\section{IP and CP models}\label{sec:cp_ip}

\subsection{Integer Programming formulations}

\subsubsection{Description of the models}\label{sec:ip_model}
In this section we present the integer programming formulations for FC-SRI and egalitarian, generous, rank-maximal and almost-stable SRI.

We wrote Rothblum's \cite{rothblum_characterization_1992} formulation for SMI in terms of an undirected graph $(V,E)$, like in \cite{abeledo_stable_1994}. 

Let $\{v, u\} \in E$ if $u$ and $v$ find each other acceptable. If $u, v \in V$ are matched, $x_{u,v} =1$, otherwise $x_{u,v} =0$. $N(v)$ denotes the set of neighbours of the vertex $v$, i.e. all the agents $v$ finds acceptable and who find $v$ acceptable.
\begin{align}
&\sum_{u \in N(v)}x_{u, v} \leq 1 \textit{ for each } v \in V  \\
&x_{u, v} = 0 \textit{ for each } \{u, v\} \notin E \\
&\sum_{i \in \{ N(u): i >_u v\}}x_{u, i} + \sum_{j \in \{ N(v): i >_v u\}}x_{v, j} + x_{u, v}\geq 1  \nonumber \\
&\qquad\qquad\textit{ for each } \{u, v\} \in E \\
&x_{u, v} \in \{0, 1\} \textit{ for each } \{u, v\} \in E \\
&x_{u,v} = x_{v,u} \textit{ for each } u,v \in V
\end{align}

The first condition enforces that no agent is matched to more than one other agent. The second condition enforces that if two agents find each other unacceptable, they are not matched. The third condition states that if $u$ and $v$ find each other acceptable, either $u$ is matched to someone they prefer to $v$, $v$ is matched to someone they prefer to $u$, or $u$ and $v$ are matched. This means there are no blocking pairs. The fourth statement restricts that agents can only be either matched or unmatched. The last condition enforces that if $u$ is matched to $v$, $v$ is also matched to $u$.

To get the egalitarian stable matching, we want to minimise 
$\sum_{u, v \in V} rank(u, v)x_{u,v}$.

For the profile-based optimality criteria, let us first define the function 
\begin{equation*}
\delta^r(u, v) = \begin{cases}
    1 \quad rank(u,v) = r \\
    0 \quad \textit{otherwise}
\end{cases}
\end{equation*}

For FC-SRI, we want to maximise $\sum_{u, v \in V} \delta^1(u,v)x_{u,v}$. This will maximise the number of first choices.

For rank-maximal, we have to solve multiple optimisation problems.

In the first iteration, we solve for the FC-SRI and then set $y^1$ to be the number of first choices in a first-choice-maximal matching. In each subsequent iteration $i$, we have constraints $\sum_{u, v} x_{u,v}\delta^j(u,v) \geq y^j$ for all $\{u, v\} \in E$, $1 \leq j \leq i - 1$, and we maximise the number of $i^{th}$ choices: $\sum_{u, v \in V} x_{u,v}\delta^{i}(u,v)$. We define $y^i$ as the solution to that maximisation problem and use it in the following iterations. The last iteration we run is $i=L-1$, where $L$ is the length of the longest preference list an agent has.

For the generous stable matching, let us define $\Tilde{y}^L$ as minimising $\sum_{u, v \in V} x_{u,v}\delta^{L}(u,v)$, subject to the original constraints. In each subsequent iteration $i$, we have constraints $\sum_{u, v} x_{u,v}\delta^j(u,v) \leq \Tilde{y}^j$ for all $\{u, v\} \in E$, $i + 1 \leq j \leq L$. We minimise the number of $i^{th}$ choices: $\sum_{u, v \in V} x_{u,v}\delta^{i}(u,v)$. The last iteration we run is $i=2$. 

To model the almost stable matching, we want to be able to count the number of blocking pairs. We add a variable $b_{u,v}$ to the third condition, and constrain $b_{u,v}$ to the interval $\{0,1\}$:
\begin{align*}
&\sum_{i \in \{ N(u): i >_u v\}}x_{u, i} + \sum_{j \in \{ N(v): i >_v u\}}x_{v, j} + x_{u, v} + b_{u, v} \geq 1 \\
&\qquad\qquad\textit{ for each } \{u, v\} \in E\\
&b_{u, v} \in \{0, 1\} \textit{ for each } \{u, v\} \in E
\end{align*}
Now our condition states that either $u$ is matched to someone they prefer to $v$, $v$ is matched to someone they prefer to $u$, $u$ and $v$ are matched or $b_{u,v}=1$. Now if $u$ and $v$ form a blocking pair, we must have $b_{u,v}=1$. It follows that minimising  $\sum_{\{u,v\}\in E}b_{u,v}$ minimises the number of blocking pairs.

\subsubsection{Integer Programming implementation}\label{sec:ip_imp}
The formulation was implemented using Gurobi's \cite{gurobi} Python API. We chose Gurobi, because it is among the most powerful solvers \cite{anand_comparative_2017}, is free for academic use and has many accessible APIs. We chose the Python API because of its simplicity and ease of use. The source code is available at \cite{github_repo}.

The following lines create the basic SRI model:
\begin{lstlisting}[language=Python,backgroundcolor = \color{white}, numbers=right, breaklines]
    x = m.addVars(n, n, vtype=GRB.BINARY)
    m.addConstrs(x.sum([u for u in h.get_neighbours(v)], v) <= 1 for v in range(n))
    m.addConstrs(x.sum(u, v) == 0 for u,v in h.get_non_edges())
    m.addConstrs(
      x.sum(u, [i for i in h.get_preferred_neighbours(u,v)])
      + x.sum([i for i in h.get_preferred_neighbours(v,u)], v) 
      + x[u, v] >= 1
                for u,v in h.get_edges())
    m.addConstrs(x[u,v] == x[v,u] for u in range(n) for v in range(n))
\end{lstlisting}
The first line creates the variables $x_{u,v}$ and constraint 4. The second line adds constraint 1. The third line adds constraint 2. Lines 4-8 add constraint 5.

We created a class \lstinline{PreferenceHelper} that reads in a file and that supplies us with helpful functions related to the structure of the preferences. In the code above, \lstinline{h} is the\lstinline{PreferenceHelper} object. \lstinline{h.get_edges()} returns the set of acceptable pairs $E$ and \lstinline{h.get_non_edges()} the set of non-acceptable pairs. \lstinline{h.get_neighbours(v)} gives the neighbours of v, i.e. $N(v)$. \lstinline{h.get_preferred_neighbours(u,v)} gets the subset of $N(u)$ that $u$ prefers to $v$.

The optimality criteria are set using \lstinline{m.setObjective} function. \lstinline{PreferenceHelper} has a helper function \lstinline{h.delta(i)}, which builds a dictionary from $(u,v)$ to $\delta^i(u,v)$ and the attribute \lstinline{ranks}, which is a dictionary from $(u,v)$ to $rank(u,v)$. These are used to pass the functions $rank$ and $\delta^i$ to Gurobi.

A detailed description of the implementation of the different optimality criteria is in Appendix~\ref{app:ip_opt}. They all closely follow the descriptions outlined in Section~\ref{sec:ip_model}. It is important to note that the IP model for generous and rank-maximal SRI reuses the same model every iteration, but it does not take into account any information we might have received from the previous iteration. 

We used the following Gurobi-parameters:
\begin{itemize}
\setlength\itemsep{0em}
\item Egalitarian: BranchDir = -1, Heuristics = 0, PrePasses = 2
\item FC-SRI: Heuristics = 0, MIPFocus = 3, NoRelHeurWork = 60
\item Rank-maximal: ScaleFlag  = 1, Heuristics = 0, BranchDir = 1, PrePasses = 5, MIPFocus = 3
\item Generous: DegenMoves = 4, Heuristics = 0, PrePasses = 5, BranchDir = -1, MIPFocus = 2
\item Almost-stable: NormAdjust = 0,  Heuristics = 0.001, VarBranch = 1, GomoryPasses = 15, PreSparsify = 0
\end{itemize}

These were found using Gurobi's parameter-tuning tool.

\subsection{Constraint Programming formulations}

\subsubsection{Description of the models}

Our CP formulation is based on Prosser's simple formulation from \cite{prosser_constraint_2014}. The optimality criteria are added by us.

The main decision variable of the model is $agent[i], 0 \leq i < n$, where $n$ is the number of agents. The value of $agent[i]$ represent the rank of the agent $i$ is matched to. 

It has the constraints

\[ agent[i] > rank(i,j) \implies agent[j] < rank(j,i)\]
\[ agent[i] = rank(i, k) \implies agent[j] = rank(j,i)\]
where $0 \leq i < n$, $0 \leq j < n$.

The first one implies there are no blocking pairs: we cannot have a situation where both $i$ and $j$ are matched to someone they like less than each other. The second one implies that if $i$ is matched to $j$, then $j$ is also matched to $i$.

Additionally, when reading the instances, each agents' last choice is themselves. If an agent is matched to themselves, we treat the agent as unmatched.
\subsubsection{Egalitarian}
For the egalitarian SRI we also have the variable $cost$, which has the constraint
\[ \sum_{i = 0}^{n-1} agent[i] = cost \]
Our objective is to minimise the cost.

\subsubsection{Profile-based optimality}
For generous, rank-maximal and first-choice-maximal we have the variable array $profile$ with the constraint
\[profile[i] = | \{agent[j] = i | 0 \leq j < n\} |, 0 \leq i < n\]

For first-choice-maximal SRI we maximise $profile[0]$.

For rank-maximal, we solve iteratively. We have the constraint
\[ profile[i] \geq minprofile[i], \forall i \in 0 \dots n-1\]
where $minprofile[i]$ is initially 0 for every $i$. During the $i^{th}$ iteration we maximise $profile[i]$ and then update $minprofile[i]$ to be the optimal value obtained for $profile[i]$. This approach is similar to the one used for the IP formulation.

We use a similar approach for generous SRI. We have the constraint
\[ profile[i] \leq maxprofile[i], \forall i \in 0 \dots n-1\]
where maxprofile is initially $n$ for every $i$. In the $i^{th}$ iteration we minimise $profile[n - 1 - i]$ and then update $maxprofile[n - 1 - i]$ to be the optimal value obtained for $profile[n - 1- i]$.

\subsubsection{Almost Stable}
For almost stable SRI we add a $n \times n$ binary array $blocking$. The constraint that enforces there are no blocking pairs is replaced with the constraint
\begin{eqnarray*}
agent[i] > rank(i,j) \implies  \\
(agent[j] < rank(j,i) \vee blocking[i,j] = 1) 
\end{eqnarray*}
This enforces that either the agents $i$ and $j$ are not a blocking pair, or the variable $blocking[i,j]$ is 1. 
We also have variable $blocking\_sum$ which is the sum of all $blocking[i,j]$. Therefore, this variable counts the blocking pairs. More formally,
\[blocking\_sum = \sum_{i = 0}^{n-1} \sum_{j = 0}^{n-1} blocking[i,j]\]
We minimise for the $blocking\_sum$.

\subsubsection{Implementation}\label{sec:cp_imp}
We translated Prosser's implementation from Choco \cite{prud2016choco} 2.1.5 to Choco 4.0.8 and added the option of computing the rank-maximal, generous, first-choice-maximal, almost-stable and egalitarian SRI as described above. The code is available at \cite{github_repo}. We updated the code because the documentation for Choco 2.1.5 was no longer easily available.

The basic SRI model is set up with the same logic as in Prosser \cite{prosser_constraint_2014}. The model is created with \lstinline{model = new Model();} and the agents with \lstinline{agent = model.intVarArray("agents", n, 0, n - 1);}. The argument 0 sets the minimum value, the argument $n-1$ the maximum value.

As an example of the constraints, the constraint $agent[i] > rank(i,j) \implies agent[j] < rank(j,i)$ is implemented with 
\begin{lstlisting}[backgroundcolor = \color{white},  breaklines]
model.ifThen(model.arithm(agent[i], ">", rank[i][k]), model.arithm(agent[k], "<", rank[k][i]));
\end{lstlisting}
\lstinline{model.arithm(agent[i], ">", rank[i][k]} creates an arithmetic constraint, that states that the rank of agent $i$'s partner is higher than $rank(i, k)$. \lstinline{model.ifThen(constraintA, constraintB)} creates an implication \lstinline{A} $\implies$ \lstinline{B} and enforces it in the model.

We use egalitarian to give an example of the optimality criteria. We create an IntVar \\ \lstinline{cost = model.intVar("cost", 0, n * (n - 1));} and set our objective to minimising that with \\ \lstinline{model.setObjective(Model.MINIMIZE, cost);}. It is constrained to be the sum of the ranks of the agents: \lstinline{model.sum(agent, "=", cost).post();}

For generous and rank-maximal SRI we recreate the model for every iteration and add the new constraint. A more sophisticated implementation could save and reuse the model.

\section{CP and IP model evaluation}\label{sec:cp_ip_eval}
\begin{table}[ht]
\centering
\begin{tabular}{|l|l|l|l|l|l|}
\hline
completeness (\%) & size & IP & CP & ASP & AF\\
\hline
25 & 20 & 0.035 & 0.026 &  & \\
 & 40 & 0.142 & 0.012 & 0.027 & \\
 & 60 & 0.354 & 0.020 & 0.073 & \\
 & 80 & 0.697 & 0.037 & 0.173 & 0.121\\
 & 100 & 1.201 & 0.066 & 0.392 & 0.246\\
 & 150 & 3.404 & 0.171 & 1.994 & 1.047\\
 & 200 & 7.393 & 0.348 & 8.559 & 2.639\\
\hline
50 & 20 & 0.045 & 0.013 &  & \\
 & 40 & 0.235 & 0.019 & 0.071 & \\
 & 60 & 0.670 & 0.048 & 0.297 & \\
 & 80 & 1.457 & 0.089 & 0.941 & 0.523\\
 & 100 & 2.713 & 0.161 & 3.046 & 1.048\\
 & 150 & 8.991 & 0.470 & 15.793 & 4.996\\
 & 200 & 21.562 & 0.930 & 70.987 & 13.812\\
\hline
75 & 20 & 0.064 & 0.012 &  & \\
 & 40 & 0.386 & 0.032 & 0.157 & \\
 & 60 & 1.178 & 0.081 & 0.833 & \\
 & 80 & 2.770 & 0.171 & 3.596 & 1.250\\
 & 100 & 5.433 & 0.277 & 10.269 & 2.787\\
 & 150 & 18.987 & 0.872 & 51.169 & 12.415\\
 & 200 & 47.124 & 2.388 & 186.470 & 36.098\\
\hline
100 & 20 & 0.088 & 0.015 &  & \\
 & 40 & 0.584 & 0.046 & 0.245 & \\
 & 60 & 1.909 & 0.118 & 1.671 & \\
 & 80 & 4.637 & 0.236 & 7.456 & 2.524\\
 & 100 & 9.265 & 0.442 & 20.349 & 5.578\\
 & 150 & 33.618 & 1.279 & 112.671 & 24.855\\
 & 200 & 86.133 & 3.300 & 374.668 & 76.602\\
\hline
\end{tabular}
\caption{The performance of IP, CP, the formulation by Erdam et. al.\cite{erdem_answer_2020} (ASP) and the formulation by Amendola \cite{amendola_solving_2018} (AF) for the general SRI. The blank entries indicate values that were not present in the source paper.}\label{fig:performance_general}
\end{table}

\begin{table}[ht]
\centering
\begin{tabular}{|l|l|l|l|l|}
\hline
completeness (\%) & size & IP & CP & ASP\\
\hline
25 & 20 & 0.054 & 0.029 & \\
 & 40 & 0.170 & 0.013 & 0.204\\
 & 60 & 0.443 & 0.028 & 0.129\\
 & 80 & 0.849 & 0.046 & 0.255\\
 & 100 & 1.443 & 0.077 & 0.575\\
 & 150 & 3.357 & 0.320 & 14.126\\
 & 200 & 7.254 & 0.433 & 59.18\\
\hline
50 & 20 & 0.060 & 0.013 & \\
 & 40 & 0.312 & 0.026 & 0.106\\
 & 60 & 0.805 & 0.058 & 0.535\\
 & 80 & 1.832 & 0.112 & 1.82\\
 & 100 & 3.339 & 0.197 & 4.97\\
 & 150 & 8.839 & 0.600 & 153.54\\
 & 200 & 20.866 & 1.160 & 524.58\\
\hline
75 & 20 & 0.068 & 0.015 & \\
 & 40 & 0.434 & 0.040 & 0.722\\
 & 60 & 1.518 & 0.104 & 8.171\\
 & 80 & 3.573 & 0.193 & 6.99\\
 & 100 & 6.767 & 0.336 & 19.04\\
 & 150 & 18.291 & 0.934 & 492.7\\
 & 200 & 44.435 & 2.958 & 1757.0\\
\hline
100 & 20 & 0.107 & 0.015 & \\
 & 40 & 0.727 & 0.062 & 0.463\\
 & 60 & 2.450 & 0.158 & 2.534\\
 & 80 & 5.895 & 0.325 & 14.39\\
 & 100 & 11.543 & 0.500 & 35.92\\
 & 150 & 31.941 & 1.604 & 360.7\\
 & 200 & 79.483 & 4.298 & 844.03\\
\hline
\end{tabular}
\caption{The performance of IP, CP and ASP \cite{erdem_answer_2020} models in seconds for the egalitarian SRI. The blank entries indicate values that were not present in the source paper.}\label{fig:egalitarian_sri_performance}
\end{table}

\begin{table}[ht]
\centering
\begin{tabular}{|l|l|l|l|}
\hline
completeness (\%) & size & IP & CP\\
\hline
25 & 20 & 0.035 & 0.032\\
 & 40 & 0.141 & 0.016\\
 & 60 & 0.354 & 0.027\\
 & 80 & 0.695 & 0.055\\
 & 100 & 1.199 & 0.090\\
 & 150 & 3.408 & 0.219\\
 & 200 & 7.422 & 0.459\\
\hline
50 & 20 & 0.045 & 0.017\\
 & 40 & 0.233 & 0.030\\
 & 60 & 0.670 & 0.059\\
 & 80 & 1.444 & 0.118\\
 & 100 & 2.701 & 0.203\\
 & 150 & 9.000 & 0.568\\
 & 200 & 21.242 & 1.167\\
\hline
75 & 20 & 0.064 & 0.016\\
 & 40 & 0.385 & 0.048\\
 & 60 & 1.180 & 0.109\\
 & 80 & 2.773 & 0.221\\
 & 100 & 5.440 & 0.348\\
 & 150 & 18.627 & 1.049\\
 & 200 & 45.351 & 2.947\\
\hline
100 & 20 & 0.089 & 0.013\\
 & 40 & 0.585 & 0.060\\
 & 60 & 1.928 & 0.151\\
 & 80 & 4.661 & 0.308\\
 & 100 & 9.230 & 0.563\\
 & 150 & 32.493 & 2.001\\
 & 200 & 81.167 & 4.068\\
\hline
\end{tabular}
\caption{The performance of IP and CP models in seconds for the first-choice-maximal SRI}\label{fig:1stmax-performance}
\end{table}

\begin{table}[ht]
\centering
\begin{tabular}{|l|l|l|l|l|}
\hline
completeness (\%) & size & IP & CP & ASP\\
\hline
25 & 20 & 0.061 & 0.108 & \\
 & 40 & 0.240 & 0.167 & 0.201\\
 & 60 & 0.612 & 0.707 & 0.219\\
 & 80 & 1.325 & 2.485 & 0.256\\
 & 100 & 2.835 & 5.717 & 0.602\\
 & 150 & 9.219 & 14.835 & 17.703\\
 & 200 & 37.436 & 53.609 & 83.85\\
\hline
50 & 20 & 0.086 & 0.079 & \\
 & 40 & 0.488 & 0.469 & 0.108\\
 & 60 & 2.555 & 2.701 & 0.56\\
 & 80 & 5.231 & 6.065 & 1.85\\
 & 100 & 13.176 & 12.300 & 5.26\\
 & 150 & 103.559 & 72.184 & 149.447\\
 & 200 & 251.349 & 167.629 & 704.3\\
\hline
75 & 20 & 0.136 & 0.100 & \\
 & 40 & 1.268 & 0.925 & 0.33\\
 & 60 & 6.278 & 3.817 & 1.85\\
 & 80 & 12.213 & 6.502 & 7.16\\
 & 100 & 52.466 & 22.996 & 20.65\\
 & 150 & 228.700 & 91.175 & 529.16\\
 & 200 & 1135.815 & 362.521 & 1475.0\\
\hline
100 & 20 & 0.194 & 0.113 & \\
 & 40 & 2.383 & 1.542 & 0.56\\
 & 60 & 14.846 & 6.237 & 3.372\\
 & 80 & 42.392 & 15.437 & 15.47\\
 & 100 & 90.961 & 27.413 & 40.42\\
 & 150 & 651.276 & 190.827 & 362.12\\
 & 200 & TO & 530.909 & 1147.0\\
\hline
\end{tabular}
\caption{The performance of IP, CP and ASP \cite{erdem_answer_2020} models in seconds for the rank-maximal SRI. The timeout is 3000 seconds. If a single instance times out, we treat the average of instances as timed out. The blank entries indicate values that were not present in the source paper.}\label{fig:rankmax_performance}
\end{table}

\begin{table}[ht]
\centering
\begin{tabular}{|l|l|l|l|}
\hline
completeness (\%) & size & IP & CP\\
\hline
25 & 20 & 0.044 & 0.156\\
 & 40 & 0.162 & 0.278\\
 & 60 & 0.432 & 0.963\\
 & 80 & 0.918 & 3.378\\
 & 100 & 1.689 & 7.684\\
 & 150 & 4.918 & 20.410\\
 & 200 & 13.736 & 50.218\\
\hline
50 & 20 & 0.058 & 0.116\\
 & 40 & 0.289 & 0.636\\
 & 60 & 0.950 & 3.474\\
 & 80 & 2.118 & 7.435\\
 & 100 & 4.284 & 15.273\\
 & 150 & 20.625 & 76.575\\
 & 200 & 43.300 & 146.388\\
\hline
75 & 20 & 0.087 & 0.155\\
 & 40 & 0.521 & 1.196\\
 & 60 & 1.617 & 4.535\\
 & 80 & 3.890 & 7.775\\
 & 100 & 10.273 & 28.178\\
 & 150 & 36.575 & 93.860\\
 & 200 & 107.884 & 342.794\\
\hline
100 & 20 & 0.112 & 0.146\\
 & 40 & 0.794 & 1.817\\
 & 60 & 2.814 & 6.854\\
 & 80 & 8.163 & 17.646\\
 & 100 & 16.451 & 31.880\\
 & 150 & 69.062 & 171.903\\
 & 200 & 215.839 & 520.349\\
\hline
\end{tabular}
\caption{The performance of IP and CP models in seconds for the generous SRI}\label{fig:generous_performance}
\end{table}

\begin{table}[ht]
\centering
\begin{tabular}{|l|l|l|l|l|}
\hline
completeness (\%) & size & IP & CP & ASP\\
\hline
25 & 20 & 0.061 & 0.074 & \\
 & 40 & 0.194 & TO & 0.017\\
 & 60 & 0.484 & TO & 0.479\\
 & 80 & 0.989 & TO & 4.416\\
 & 100 & 1.816 & TO & 85.86\\
 & 150 & 4.623 & TO & TO\\
 & 200 & 10.922 & TO & TO\\
\hline
50 & 20 & 0.064 & 0.971 & \\
 & 40 & 0.354 & TO & 0.136\\
 & 60 & 0.993 & TO & 3.748\\
 & 80 & 2.282 & TO & 3.748\\
 & 100 & 4.358 & TO & 343.24\\
 & 150 & 12.572 & TO & TO\\
 & 200 & 44.174 & TO & TO\\
\hline
75 & 20 & 0.078 & 3.462 & \\
 & 40 & 0.618 & TO & 0.486\\
 & 60 & 1.909 & TO & 15.526\\
 & 80 & 4.491 & TO & 144.80\\
 & 100 & 8.978 & TO & 885.39\\
 & 150 & 27.018 & TO & TO\\
 & 200 & 83.185 & TO & TO\\
\hline
100 & 20 & 0.122 & 9.115 & \\
 & 40 & 0.914 & TO & 0.675\\
 & 60 & 3.078 & TO & 27.935\\
 & 80 & 7.495 & TO & 181.09\\
 & 100 & 15.998 & TO & 2627.77\\
 & 150 & 51.163 & TO & TO\\
 & 200 & 146.761 & TO & TO\\
\hline
\end{tabular}
\caption{The performance of IP, CP and ASP \cite{erdem_answer_2020} models in seconds for the almost-stable SRI. The timeout is 3000 seconds. If any of the instances times out, the average is TO. The blank entries were not present in the source paper.}\label{fig:almost_performance}
\end{table}

We evaluated our models against the random instances generated for \cite{erdem_answer_2020}. We had instances of size 20, 40, 60, 80, 100, 150 and 200. For each size, we had instances with completeness probability 25\%, 50\%, 75\% and 100\%. This is the probability that two agents find each other acceptable. For each size-instance combination we had 20 instances. Our timeout was 3000 seconds, same as in \cite{erdem_answer_2020}.

\subsection{Correctness}
We compared the outputs of CP and IP models against each other, and where applicable, against the outputs of the ASP model from \cite{erdem_answer_2020}. In most cases we either got identical matchings or matchings which received the same measure relative to the optimality criteria in question.

ASP claims the instance size:20, density:25, number:11 does not admit a stable matching but CP and IP models find answers. We suspect this is some kind of error with running the ASP.

Additionally, for 21 instances (all of size 40, density 25 and the instance 1 of size 60, density 25) IP answer had fewer blocking pairs than the ASP formulation. Our stability checker found  more blocking pairs from the ASP matching than the output claimed there were, generally tens of them. We do not know whether this was an error in the ASP model, in running it or whether the output file we have is referring to different instances. In all the other cases, ASP and IP models reported the same number of blocking pairs.

\subsection{Performance}
We evaluated the performance of the CP and IP models against each other and the ASP models from Amendola \cite{amendola_solving_2018} and Erdem et. al. \cite{erdem_answer_2020} (note that the performance evaluation of the former is in the latter). The Amendola's ASP formulation only computes the general SRI, whereas Erdem's one also solves for egalitarian, rank-maximal and almost-stable SRI. The ASP formulations have been run on a different machine, so the values are not directly comparable. Their magnitude and speed of increase can still give us an idea on how well the models compare.

We ran the models on a machine with Intel Core i3 CPU 530 @ 2.93GHz processor and 8.09 GB of memory. The runtimes are presented on Tables~\ref{fig:performance_general}-\ref{fig:almost_performance}. We see that CP outperforms IP on most optimality criteria, except almost-stable and to a lesser extent generous. Both solvers mostly outperform the previous ASP formulations. We describe the results in more detail in the following sections.

\subsubsection{General SRI}
The runtimes are as in Table~\ref{fig:performance_general}. CP remains the fastest solver for the general SRI by a wide margin. When the other solvers take tens of seconds for the largest instances, CP solves those within a few seconds. IP outperformed Erdem's ASP formulation and achieved runtimes similar to Amendola's ASP model.

\subsubsection{Egalitarian and First-Choice-Maximal SRI}
The average total times for building and solving the models are as in Table~\ref{fig:egalitarian_sri_performance} and Table~\ref{fig:1stmax-performance}.
Both CP and IP models were approximately as fast with this problem as they were with the general SRI. Consequently, CP outperformed IP by a wide margin. IP was about 10 times faster than Erdem's ASP formulation in the most difficult instances, CP was about 200 times faster.

\subsubsection{Generous and Rank-Maximal}
The runtimes are as in Table~\ref{fig:rankmax_performance} and Table~\ref{fig:generous_performance}.
Generous and rank-maximal SRI were substantially more computationally expensive than first-choice-maximal and egalitarian SRI.

CP still generally outperforms IP on the rank-maximal SRI, but the factors are much less drastic - CP is about 5 times faster than ASP. Additionally, on the instances with completeness 25\%, IP is actually outperforming CP. The performance of IP and ASP is relatively even, with IP outperforming ASP on big instances with completeness 25 - 75\% and ASP outperforming IP when the completeness is 100\%.

Curiously, IP was noticeable faster at solving generous SRI than CP and in fact outperformed CP there. It is likely that a more sophisticated CP implementation would perform better (see Section~\ref{sec:cp_imp}) although there might also be space to improve the IP implementation (see Section~\ref{sec:ip_imp}).

\subsubsection{Almost-stable}
Almost-stable was the only optimality criteria where CP had serious difficulties. The formulation started timing out already in instances of size 40. We speculate this is because CP is better suited for problems with few feasible solutions. For the almost-stable SRI, every matching is feasible, not just the stable ones. This massively increases the search space and decreases the inferences a CP solver can make.

The IP model, on the other hand, performed well. It managed to solve for problem instances where even the ASP model timed out. Interestingly, IP actually solved almost-stable SRI faster than rank-maximal and generous SRI. Since rank-maximal and generous should be theoretically easier - there are far fewer stable matchings than all possible matchings - these formulations could probably be improved.

As an interesting side note, most of the almost-stable matchings did not contain many blocking pairs. Most instances admitted a matching with a single blocking pair, a few required two and one as many as four.

\section{Conclusions and future work}\label{sec:future}

We have proven that the rank-maximal and generous SRI and FC-SRI remain NP-hard even when the preference lists are of length at most 3. Although not surprising, these results mean that we cannot use short-preference lists to find polynomial-time solvable special cases of the problems.

We showed that the number of $k^{th}$ choices, where $k$ is the minimum-regret, can be 2-approximated. This on its own is not a 2-approximation algorithm for the generous SRI, but it gives us hope we might find one. On the other hand, we showed that the number of first choices does not admit an approximation algorithm with a constant performance guarantee. Therefore it is very unlikely that we can find a meaningful approximation algorithm for the rank-maximal SRI.

We presented and evaluated IP formulations for SRI and its different optimality criteria. Additionally, we added the different optimality criteria to Prosser's CP formulation \cite{prosser_constraint_2014}.

A few interesting questions remain open. Does generous SRI admit an approximation algorithm? It is not even clear how to define measure on profiles. One approach would be the restriction $\sum_{i=0}^{n-1} c p^r(M_{opt})[i] \geq \sum_{i=0}^{n-1} p^r(M)[i]$, $0 \leq i < n$ where $c$ is the performance guarantee.

It is also unknown whether FC-SRI remains c-inapproximable if the preference lists are of a bounded length. The reduction we presented does not bound the length of agents' preference lists.

Despite the number of first-choices being inapproximable, there might be a related measure that could be approximated. The reduction had many first choices which could never be attained. Could we maximise the number of people who get the best agent they can get in any stable matching?

The rank-maximal and generous CP and IP models could probably be improved. Especially CP both solves and builds for the same instance of SRI multiple times without remembering anything it has computed so far. IP builds the model only once, but does not use information from the previous iterations of solving it. Additionally, it is possible that solving multiple optimisation problems is not the most efficient way to solve for rank-maximal and generous SRI. An approach that uses exponential weights might work better, but one needs to make sure that the solver supports them.

\vskip8pt \noindent
{\bf Acknowledgments.}
The authors would like to thank Manuel Sorge and Jiehua Chen for observing that the inapproximability reduction for FC-SRI can also be used to show W[1]-hardness. Thank you for Dr. Jiehua Chen for also coming up with the idea for the XP-algorithm.

The authors would like to thank Alice Ravier for her help with proofreading and proofchecking. I am genuinely impressed by how many inequalities were fixed thanks to her.

\bibliographystyle{abbrv}
\bibliography{example}

\FloatBarrier

\newpage
\newpage
\FloatBarrier
\appendix

\section{Examples of the different optimality criteria}\label{app:optimality_examples}

\begin{table}[ht]
\centering
\begin{tabular}{ c c c c c c c c c c} 
$a_1$: & $a_8$ & $a_2$ & $a_9$& \cellcolor{green} $a_3$ & $a_6$ & $a_4$& $a_5$ & $a_7$ & $a_{10}$ \\ 
$a_2$: & \cellcolor[RGB]{255, 0 ,0} $a_4$ & $a_3$ & $a_8$& $a_9$ & $a_5$ & $a_1$& $a_{10}$ & $a_6$ & $a_7$ \\ 
$a_3$:  & $a_5$ & $a_6$ & $a_8$& $a_2$ & \cellcolor{green} $a_1$ & $a_7$& $a_{10}$ & $a_4$ & $a_9$ \\  
$a_4$: & $a_{10}$ & $a_7$ & $a_9$& $a_3$ & $a_1$ & $a_6$& \cellcolor[RGB]{255, 0 ,0}$a_2$ & $a_5$ & $a_8$ \\ 
$a_5$: & \cellcolor[RGB]{100, 150, 255}$a_7$ & $a_4$ & $a_{10}$& $a_8$ & $a_2$ & $a_6$& $a_3$ & $a_1$ & $a_9$ \\ 
$a_6$:  & $a_2$ & \cellcolor[RGB]{255, 255, 0} $a_8$ & $a_7$& $a_3$ & $a_4$ & $a_{10}$& $a_1$ & $a_5$ & $a_9$ \\ 
$a_7$: & $a_2$ & $a_1$ & $a_8$& $a_3$ & \cellcolor[RGB]{100, 150, 255}$a_5$ & $a_{10}$& $a_4$ & $a_6$ & $a_9$ \\ 
$a_8$: & $a_{10}$ & $a_4$ & $a_2$& $a_5$ & \cellcolor[RGB]{255, 255, 0} $a_6$ & $a_7$& $a_1$ & $a_3$ & $a_9$ \\ 
$a_9$:  & $a_6$ & $a_7$ & $a_2$& $a_5$ & \cellcolor[RGB]{255, 100, 0} $a_{10}$ & $a_3$& $a_4$ & $a_8$ & $a_1$ \\  
$a_{10}$:  & $a_3$ & $a_1$ & $a_6$& $a_5$ & $a_2$ & \cellcolor[RGB]{255, 100, 0} $a_9$& $a_8$ & $a_4$ & $a_7$ \\  
\end{tabular}
\caption{An instance of SR with seven stable matchings, with one of them highlighted. \cite{gale_college_1962}}
\label{table:sr_stable}
\end{table}

\begin{table}[ht]
\centering
\begin{tabular}{ c} 
$R_1$ = \{\{1, 3\}, \{2, 4\}, \{5, 7\}, \{6, 8\}, \{9, 10\}\}  \\
$R_2$ = \{\{1, 7\}, \{2, 8\}, \{3, 5\}, \{4, 9\}, \{6, 10\}\}  \\ 
$R_3$ = \{\{1, 4\}, \{2, 9\}, \{3, 6\}, \{5, 7\}, \{8, 10\}\}  \\ 
$R_4$ = \{\{1, 4\}, \{2, 3\}, \{5, 7\}, \{6, 8\}, \{9, 10\}\}  \\ 
$R_5$ = \{\{1, 4\}, \{2, 8\}, \{3, 6\}, \{5, 7\}, \{9, 10\}\}  \\ 
$R_6$ = \{\{1, 7\}, \{2, 3\}, \{4, 9\}, \{5, 10\}, \{6, 8\}\}  \\ 
$R_7$ = \{\{1, 7\}, \{2, 8\}, \{3, 6\}, \{4, 9\}, \{5, 10\}\}  \\ 
\end{tabular}
\caption{The stable matching of the instance in \ref{table:sr_stable} \cite{gale_college_1962}}
\label{table:sr_stable_matchings}
\end{table}

\begin{table}[ht]
\centering
\begin{tabular}{ c} 
$p(R_1) = \langle 2, 1, 0, 1, 4, 1, 1, 0, 0 \rangle $ \\
$p(R_2) = \langle 1, 1, 4, 0, 0, 1, 2, 1, 0 \rangle $ \\
$p(R_3) = \langle 2, 1, 1, 2, 2, 1, 1, 0, 0 \rangle $ \\
$p(R_4) = \langle 1, 2, 0, 1, 4, 2, 0, 0, 0 \rangle $ \\
$p(R_5) = \langle 1, 1, 2, 1, 3, 2, 0, 0, 0 \rangle $ \\
$p(R_6) = \langle 0, 3, 2, 2, 1, 0, 1, 1, 0 \rangle $ \\
$p(R_7) = \langle 0, 2, 4, 2, 0, 0, 1, 1, 0 \rangle $ \\
\end{tabular}
\caption{The profiles of the matchings in the table~\ref{table:sr_stable_matchings} \cite{gusfield_book_1989}}
\label{table:sr_profiles}
\end{table}

\begin{table}[ht]
\centering
\begin{tabular}{ c} 
$cost(R_1) =41 $ \\
$cost(R_2) =43 $ \\
$cost(R_3) =38 $ \\
$cost(R_4) =41 $ \\
$cost(R_5) =40 $ \\
$cost(R_6) =40 $ \\
$cost(R_7) =39 $ \\
\end{tabular}
\caption{The costs of the matchings in the table~\ref{table:sr_stable_matchings} \cite{gale_college_1962}}
\label{table:sr_costs}
\end{table}

\subsection{Egalitarian}
The matching $R_3$ is the egalitarian stable matching. This is because as we can see in Figure~\ref{table:sr_costs} it has the lowest cost.

\subsection{First-choice-maximal}
The matchings $R_1$ and $R_3$ are first choice maximal. As we can see in Figure~\ref{table:sr_profiles} both of them have 2 first choices (the first item in the profile) and no matching has more.

\subsection{Rank-maximal}
The matching $R_3$ is the rank-maximal stable matching. As we can see in Figure~\ref{table:sr_profiles}, both it and $R_1$ have the maximum number of first choices and they are tied on the number of second choices. We can see that $R_3$ has more third choices, though.

\subsection{Minimum-regret}
The matchings $R_4$ and $R_5$ are minimum-regret stable matchings. As we can see in Figure~\ref{table:sr_profiles} both of them have someone ranked to their sixth choice and no one lower than that. No other matching does better: all the other matching match someone to their seventh choice, some even to their eight.

\subsection{Generous}
The matching $R_5$ is the generous stable matching. We can see both it and $R_4$ have no ninth, eight or seventh choices. They both also have 2 sixth choices. The difference is that $R_5$ has only 3 fifth choices.

\section{IP Optimality Criteria}\label{app:ip_opt}
For egalitarian SRI we set the objective as  \lstinline{m.setObjective(x.prod(h.ranks))}. \lstinline{h.ranks} returns a dictionary from $(u,v)$ to $rank(u,v)$. \lstinline{x.prod(coeff)} is a function that multiplies each variable in $x$ with the supplied dictionary of coefficients.

For first-choice-maximal SRI we set the objective as \lstinline{x.prod(h.delta(1)}. \lstinline{h.delta(i)} returns a dictionary from $(u,v)$ pairs to $\delta^i(u, v)$.

Rank-maximal SRI is implemented as follows:
\begin{lstlisting}[language=Python,backgroundcolor = \color{white}, numbers=right]
for i in range(1, h.max_pref_length - 1):
    delta_i = h.delta(i)
    m.setObjective(x.prod(delta_i), GRB.MAXIMIZE)
    m.optimize()
    if not hasattr(m, "ObjVal"):
        has_solution = False
        break
    m.addConstr(x.prod(delta_i) >= m.ObjVal)
m.setObjective(x.prod(h.delta(h.max_pref_length - 1)), GRB.MAXIMIZE)
\end{lstlisting}
First and second line set the objective of the iteration, and the third line solves it. Lines 5-7 deal with the case where the problem is unsatisfiable. Line 8 sets the constraint for the next iteration. Line 9 solves the final optimisation problem. The code for generous SRI is very similar.

For almost stable, we create the variable $b_{u,v}$ with \lstinline{b = m.addVars(n, n, vtype=GRB.BINARY)} and add \lstinline{+ b[u,v]} to the original constraint in lines 4-8. We set our objective to \lstinline{b.sum()}.

\end{document}